\documentclass{article}
\usepackage[utf8]{inputenc}

\usepackage[margin=1in]{geometry}
\usepackage{xcolor}
\usepackage{amsmath}
\usepackage{amssymb}
\usepackage{amsthm}
\usepackage{algorithm}
\usepackage{algpseudocode}
\usepackage{hyperref}
\usepackage{graphicx}
\usepackage[font={small}]{caption}
\usepackage{thm-restate}
\usepackage[normalem]{ulem}

\usepackage{enumitem}
\usepackage{xcolor}
\newcommand{\nc}{\newcommand}
\newcommand{\DMO}{\DeclareMathOperator}

\newcommand{\noah}[1]{\textcolor{red}{[Noah: #1]}}

\newcommand{\pasin}[1]{\textcolor{red}{[Pasin: #1]}}
\newcommand{\todo}[1]{\textcolor{red}{[TODO: #1]}}
\newcommand{\modified}[1]{{#1}}
\newcommand{\remove}[1]{}

\algnotext{EndFor}
\algnotext{EndIf}
\algnotext{EndWhile}
\algnotext{EndProcedure}

\allowdisplaybreaks

\newcommand{\poly}{poly}
\nc{\MS}{\mathcal{S}}
\nc{\MP}{\mathcal{P}}
\nc{\MR}{\mathcal{R}}
\newcommand{\bM}{\mathbf{M}}
\nc{\MZ}{\mathcal{Z}}
\DMO{\Binom}{Binom}
\newcommand{\E}{\mathbb{E}}
\DMO{\Var}{Var}
\newcommand{\ba}{\mathbf{a}}
\newcommand{\bX}{\mathbf{X}}
\newcommand{\bY}{\mathbf{Y}}

\newcommand{\bt}{\mathbf{t}}
\newcommand{\bx}{\mathbf{x}}
\newcommand{\by}{\mathbf{y}}

\newcommand{\bv}{\mathbf{v}}
\newcommand{\bz}{\mathbf{z}}
\newcommand{\R}{\mathbb{R}}
\newcommand{\N}{\mathbb{N}}
\newcommand{\MX}{\mathcal{X}}

\nc{\BN}{\mathbb{N}}
\newcommand{\Z}{\mathbb{Z}}
\nc{\BZ}{\mathbb{Z}}
\newcommand{\cR}{\mathcal{R}}
\newcommand{\tDelta}{\tilde{\Delta}}
\newcommand{\bone}{\mathbf{1}}
\newcommand{\eps}{\varepsilon}
\nc{\ep}{\eps}
\newcommand{\cD}{\mathcal{D}}
\newcommand{\cF}{\mathcal{F}}
\DeclareMathOperator{\zero}{zero}
\DeclareMathOperator{\supp}{supp}

\DeclareMathOperator{\DLap}{DLap}
\DeclareMathOperator{\SD}{SD}

\DeclareMathOperator{\err}{ERR}
\renewcommand{\varepsilon}{\epsilon}

\newcommand{\DP}{\mathrm{DP}}
\newcommand{\typeOfDP}[1]{\DP_{\mathrm{ #1}}}
\newcommand{\centralDP}{\typeOfDP{central}}
\newcommand{\localDP}{\typeOfDP{local}}
\newcommand{\shuffledDP}{\typeOfDP{shuffled}}

\newtheorem{theorem}{Theorem}
\newtheorem{observation}[theorem]{Observation}
\newtheorem{lemma}[theorem]{Lemma}

\newtheorem{definition}[theorem]{Definition}
\newtheorem{question}[theorem]{Question}
\newtheorem{corollary}[theorem]{Corollary}
\newtheorem{remark}[theorem]{Remark}

\title{Pure Differentially Private Summation from Anonymous Messages}

\author{
  Badih Ghazi \hspace*{1cm}
  Noah Golowich\thanks{MIT EECS. Supported at MIT by a Fannie \& John Hertz Foundation Fellowship, an MIT Akamai Fellowship, and an NSF Graduate Fellowship.  This work was done while at Google Research.}
  \hspace*{0.5cm}
   Ravi Kumar \hspace*{0.5cm}\\ 
   Pasin Manurangsi \hspace*{0.5cm}
   Rasmus Pagh\thanks{Visiting from BARC and IT University of Copenhagen.} \hspace*{0.5cm}
   Ameya Velingker \\
   Google Research \\
   Mountain View, CA \\
   \texttt{badihghazi@gmail.com, nzg@mit.edu, ravi.k53@gmail.com,} \\
   \texttt{pasin@google.com, pagh@itu.dk, ameyav@google.com}
}

\begin{document}

\maketitle


\begin{abstract}

The \emph{shuffled} (aka \emph{anonymous}) model has recently generated significant interest as a candidate distributed privacy framework with trust assumptions better than the central model but with achievable error rates smaller than the local model. In this paper, we study \emph{pure} differentially private protocols in the shuffled model for \emph{summation}, a very basic and widely used primitive. Specifically:
\begin{itemize}
\item For the binary summation problem where each of $n$ users holds a bit as an input, we give a pure $\varepsilon$-differentially private protocol for estimating the number of ones held by the users up to an absolute error of $O_{\epsilon}(1)$, and where each user sends $O_{\eps}(\log n)$ messages each consisting of a single bit. This is the first pure protocol in the shuffled model with error $o(\sqrt{n})$ for constant values of $\epsilon$.  

Using our binary summation protocol as a building block, we give a pure $\epsilon$-differentially private protocol that performs summation of real numbers in $[0, 1]$ up to an absolute error of $O_{\epsilon}(1)$, and where each user sends $O_{\eps}(\log^{3} n)$ messages each consisting of $O(\log \log n)$ bits.

\item In contrast, we show that for any pure $\epsilon$-differentially private protocol for binary summation in the shuffled model having absolute error $n^{0.5-\Omega(1)}$, the per user communication has to be at least $\Omega_{\epsilon}(\sqrt{\log n})$ bits.  This implies (i) the first separation between the (bounded-communication) multi-message shuffled model and the central model, and (ii) the first separation between pure and approximate differentially private protocols in the shuffled model.
\end{itemize}
Interestingly, over the course of proving our lower bound, we have to consider (a generalization of) the following question that might be of independent interest: given $\gamma \in (0, 1)$, what is the smallest positive integer $m$ for which there exist two random variables $X^0$ and $X^1$ supported on $\{0, \dots, m\}$ such that (i) the total variation distance between $X^0$ and $X^1$ is at least $1 - \gamma$, and (ii) the moment generating functions of $X^0$ and $X^1$ are within a constant factor of each other everywhere? We show that the answer to this question is $m = \Theta(\sqrt{\log(1/\gamma)})$.
\end{abstract}

\newpage
\tableofcontents

\thispagestyle{empty}
\setcounter{page}{0}

\newpage

\section{Introduction}


Since its introduction by Dwork et al. \cite{dwork2006calibrating,dwork2006our}, \emph{differential privacy (DP)} has become widely popular as a rigorous mathematical definition of privacy. This has led to practical deployments at companies such as Apple~\cite{greenberg2016apple,dp2017learning}, Google~\cite{erlingsson2014rappor,CNET2014Google}, and Microsoft~\cite{ding2017collecting}, and in government agencies such as the United States Census Bureau \cite{abowd2018us}. The most widely studied setting with DP is the so-called \emph{central} model (denoted $\centralDP$) where an analyzer observes the crude user data but is supposed to release a differentially private data structure.  Many accurate private algorithms have been discovered in the central model; however, the model is limited when the analyst is not to be trusted with the user data. To remedy this,  the more appealing \emph{local} model of DP (denoted $\localDP$) \cite{kasiviswanathan2008what} (also~\cite{warner1965randomized}) requires the messages sent by each user to the analyst to be private. Nevertheless, the local model suffers from large estimation errors that are known to the be on the order of $\sqrt{n}$, where $n$ is the number of users, for a variety of problems including summation, the focus of this work~\cite{beimel2008distributed,ChanSS12}. This has motivated the study of the \emph{shuffled} model of DP (denoted $\shuffledDP$), which is intended as a middle-ground with trust assumptions better than those of the central model and estimation accuracy better than the local model.

While an analogous setup was first introduced in crytpography by Ishai et al. in their work on cryptography from anonymity~\cite{ishai2006cryptography}, the shuffled model was first proposed for privacy-preserving computations by Bittau et al.~\cite{bittau17} in their Encode-Shuffle-Analyze architecture. In this setup which is depicted in Figure~\ref{fig:esa}, each user sends (potentially several) messages to a trusted shuffler, who randomly permutes all incoming messages before passing them to the analyst. We will treat the shuffler as a black box in this work, though we point out that various efficient cryptographic implementations of the shuffler have been considered, including onion routing, mixnets, third-party servers, and secure hardware (see, e.g., the discussions in~\cite{ishai2006cryptography,bittau17}). The privacy properties of $\shuffledDP$ were first studied, independently, by Erlingsson et al.~\cite{erlingsson2019amplification} and Cheu et al.~\cite{CheuSUZZ19}. Moreover, several recent works have sought to nail down the trade-offs between accuracy, privacy and communication \cite{CheuSUZZ19,BalleBGN19,ghazi2019scalable, DBLP:journals/corr/abs-1906-09116,anon-power, ghazi2019private, balle_privacy_2019constantIKOS, balcer2019separating}.

\paragraph{Pure- and Approximate-DP.}
The two most widely used notions of DP are pure-DP~\cite{dwork2006calibrating} and approximate-DP~\cite{dwork2006our}, which we recall next. For any parameters $\epsilon \geq 0$ and $\delta \in [0,1]$, a randomized algorithm $P$ is \emph{$(\epsilon, \delta)$-DP} if for every pair datasets $X, X'$ differing on a single user's data, and for every subset $\MS$ of transcripts of $P$, it is the case that
\begin{equation}\label{eq:esp_delta_DP_intro}
\Pr[P(X) \in \MS] \leq e^\epsilon \cdot \Pr[P(X') \in \MS] + \delta,
\end{equation}
where the probabilities are taken over the randomness in $P$. The notion of $\epsilon$-DP is the special case where $\delta$ is set to $0$ in~(\ref{eq:esp_delta_DP_intro}); we use the terms \emph{pure-DP} when $\delta = 0$ and \emph{approximate-DP} when $\delta > 0$.  While $\delta$ is intuitively an upper bound on the probability that an $(\epsilon, \delta)$-DP algorithm fails to be $\epsilon$-DP, this failure event can in principle be catastrophic, revealing all the user inputs to the analyst.
Pure-DP protocols are thus highly desirable as they guarantee more stringent protections against the leakage of user data.
In the central and local settings, several prior works either obtained pure protocols in regimes where approximate protocols were previously known, or proved separations between pure and approximate protocols (e.g.,  \cite{hardt2009geometry,De12,nikolov2013geometry, steinke2015between,bun2018heavy}).



\paragraph{Summation.}
A basic primitive in data analytics and machine learning is the \emph{summation} (aka~\emph{aggregation}) of inputs held by different users. Indeed, private summation is a critical building block in the emerging area of \emph{federated learning} \cite{konevcny2016federated} (see also \cite{kairouz2019advances} for a recent extensive overview), where a machine learning model, say a neural network, is to be trained on data held by many users without having the users send their data over to a central analyzer. To do so, private variants of Stochastic Gradient Descent have been developed and their privacy/accuracy trade-offs analyzed (e.g., \cite{DPSGD}). The gist of these procedures is the private summation of users' gradient updates.
%
%
Private summation is also closely related to functions in the widely studied class of \emph{counting queries} (e.g.,~\cite{Vadhan-tutorial,blum2008learning,hardt2009geometry,hardt2010multiplicative,nikolov2013geometry}). 

Several recent work studied approximate-$\shuffledDP$ protocols for summation~\cite{CheuSUZZ19,BalleBGN19,ghazi2019scalable,DBLP:journals/corr/abs-1906-09116,ghazi2019private,balle_privacy_2019constantIKOS,balcer2019separating}. For binary summation, Cheu et al.~\cite{CheuSUZZ19} show that the standard randomized response is an $(\epsilon, \delta)$-$\shuffledDP$ protocol for binary summation and that it incurs an absolute error of only $O(\sqrt{\log n})$ 
for constant $\epsilon$ and $\delta$ inverse polynomial in $n$. For real summation in the single-message shuffled model (denoted $\shuffledDP^1$), where each user sends a single message to the shuffler, Balle et al.~\cite{BalleBGN19} show that the tight error for approximate protocols is $\Theta(n^{1/6})$. For real summation in the multi-message shuffled model (denoted $\shuffledDP^{\geq 1}$), where a user can send more than one message, the state-of-the-art approximate protocol was recently obtained in~\cite{ghazi2019private, balle_privacy_2019constantIKOS} and it incurs error at most $O(1/\epsilon)$ with every user sending $O(1 + \frac{\log(1/\delta)}{\log n})$ messages of $O(\log{n})$ bits each.

The aforementioned protocols, along with several other results (including the work on ``privacy amplification by shuffling'' of Erlingsson et al.~\cite{erlingsson2019amplification} and Balle et al.~\cite{BalleBGN19}), demonstrate the power of the shuffled model over the local model in terms of privacy, as any $(\epsilon, o(1/n))$-$\localDP$ summation protocol must incur an error of $\Omega_{\eps}(\sqrt{n})$~\cite{ChanSS12}. However, all of the protocols proposed so far in the shuffled model only achieve an advantage over the local model when allowed approximation. This leads us to the following fundamental and perplexing question that is the focus of our work:
\begin{question}
Are there pure-$\shuffledDP$ protocols that achieve better utility than any $\localDP$ protocol?
\end{question}

\subsection{Main Results}\label{subsec:our_results}

We positively answer the above question for the problem of summation.
Namely, we give the first pure-$\shuffledDP$ protocol for binary summation with error depending only on $\epsilon$ but independent of $n$ and with logarithmic communication per user.

\begin{theorem}[\bf Pure Binary Summation via Shuffling]\label{th:pure_bin_agg_protocol_informal}
	For every positive real number $\eps$, there is a (non-interactive) $\eps$-$\shuffledDP$ protocol for binary summation that has expected error $O_{\epsilon}(1)$ and where each user sends $O_{\epsilon}(\log{n})$ messages each consisting of a single bit.
\end{theorem}

We use the protocol in Theorem~\ref{th:pure_bin_agg_protocol_informal} as a building block in order to also obtain a protocol with constant error and polylogarithmic communication per user for the more general task of real summation where each user input is a real number in $[0, 1]$.

\begin{theorem}[\bf Pure Real Summation via Shuffling]
	\label{th:pure_real_agg_protocol_informal}
	For every positive real number $\epsilon$, there is a (non-interactive) $\eps$-$\shuffledDP$ protocol for real summation that has expected error $O_{\epsilon}(1)$ and where each user sends $O_{\epsilon}(\log^3 n)$ messages each consisting of $O(\log\log n)$ bits.
\end{theorem}



In light of Theorem~\ref{th:pure_bin_agg_protocol_informal}, a natural question is if there is a (non-interactive) pure-DP protocol for binary summation with logarithmic (or even constant) error and constant communication per user, as in the approximate case. We show that no such protocol exists, even for very large (polynomial) errors:
\begin{theorem}[\bf Communication Lower Bound]\label{th:pure_bin_agg_lb}
In any non-interactive $\epsilon$-DP protocol for binary summation with expected error at most $n^{0.5 - \Omega(1)}$, the worst-case per user communication must be $\Omega_{\epsilon}(\sqrt{\log{n}})$ bits.
\end{theorem}


\subsection{Implications}
Our results described above imply new separations between different types of DP protocols (e.g., $\centralDP$, $\localDP$, $\shuffledDP^1$, and $\shuffledDP^{\geq 1}$), and also give the first accurate pure-$\shuffledDP$ protocol for histograms. We elaborate on these next.

\paragraph{Pure Local vs Shuffled Protocols.}
In $\localDP$, the tight accuracy for binary summation is known to be $\Theta(\sqrt{n})$ for approximate protocols \cite{warner1965randomized, beimel2008distributed,ChanSS12}. 
Our Theorems~\ref{th:pure_bin_agg_protocol_informal} and~\ref{th:pure_real_agg_protocol_informal} give the first pure-$\shuffledDP$ protocols with error $o(\sqrt{n})$ for binary and real summation respectively, and in fact they only incur constant error for both of these problems.
Furthermore, Bun et al.~\cite{bun2018heavy} gave a generic transformation from any approximate-$\localDP$ protocol to a pure-$\localDP$ protocol with essentially the same accuracy and each user communicates only $O(\log\log{n})$ bits.  In contrast, our Theorem~\ref{th:pure_bin_agg_lb} implies that in any such transformation in the shuffled model (if one exists), the per user communication has to be $\Omega (\sqrt{\log{n}})$.

\paragraph{Pure vs Approximate Shuffled Protocols.}
Cheu et al. \cite{CheuSUZZ19} showed that the standard randomized response~\cite{warner1965randomized} is an approximate-DP protocol for binary summation that incurs only logarithmic error (for constant $\epsilon$, and $\delta$ inverse polynomial in $n$), and where each user sends a single bit.  In contrast, our Theorem~\ref{th:pure_bin_agg_lb} implies that the communication cost of any pure-DP protocol for binary summation with logarithmic error (and in fact with error as large as $n^{0.5-\Omega(1)}$) is $\Omega(\sqrt{\log{n}})$ bits. Put together, these two results imply the first separation between the communication complexity of pure-$\shuffledDP$ and approximate-$\shuffledDP$ protocols.

\paragraph{Pure Single-Message vs Multi-Message Shuffled Protocols.}
As recently shown by \cite{balcer2019separating}, any pure-$\shuffledDP^1$ protocol implies a pure-$\localDP$ protocol with the same accuracy. This implies that any pure-$\shuffledDP^1$ protocol for binary summation must incur error $\Omega_{\epsilon}(\sqrt{n})$. Our Theorem~\ref{th:pure_bin_agg_protocol_informal} thus implies a huge separation of $\Theta_{\epsilon}(\sqrt{n})$ between the errors possible for pure-$\shuffledDP^1$ and pure-$\shuffledDP^{\geq 1}$ protocols.

\paragraph{Multi-Message Shuffled vs Central Protocols.}
It is well-known that the tight error for binary summation in $\centralDP$ is $O(1/\epsilon)$ \cite{dwork2006calibrating}.
Theorem~\ref{th:pure_bin_agg_lb} proves that any $\shuffledDP$ protocol with per user communication $o_{\epsilon}(\sqrt{\log{n}})$ bits must incur error $n^{0.5-\Omega(1)}$. It thereby gives the first separation between (bounded-communication) $\shuffledDP^{\geq 1}$ and $\centralDP$ protocols. Indeed the technique used to prove Theorem~\ref{th:pure_bin_agg_lb} is, to the best of our knowledge, the first to separate the accuracy of (bounded-communication) $\shuffledDP^{\geq 1}$ protocols from those of $\centralDP$ protocols with the same privacy parameters (all previous lower bounds in the shuffled model~\cite{CheuSUZZ19,BalleBGN19,anon-power} only apply to single-message protocols).

\paragraph{Pure Protocol for Histograms.}
Our pure binary summation protocol (Theorem~\ref{th:pure_bin_agg_protocol_informal}) implies as a black-box the first pure-DP protocol with polylogarithmic error for computing \emph{histograms} (aka \emph{point functions} or \emph{frequency estimation}),  albeit with very large communication (see Appendix~\ref{sec:app_hist} for more details). It remains a very interesting open question to obtain a communication-efficient and accurate pure-DP protocol for histograms (see Section~\ref{sec:conc_oqs} for more on this and other open questions).


\subsection{Overview of Techniques}
\label{sec:technique-overview}

\paragraph{Binary Summation Protocol.}

We first explain why all existing summation protocols in the shuffled model with error $o(\sqrt{n})$ are not $O(1)$-DP. First, note that as observed by \cite{balcer2019separating}, any pure-$\shuffledDP^1$ protocol implies a pure-$\localDP$ protocol with the same accuracy and privacy. Combined with the fact that any $O(1)$-$\localDP$  protocol for summation must have error $\Omega(\sqrt{n})$, this implies the same lower bound for any pure $O(1)$-$\shuffledDP^1$ protocol. In particular, this rules out the binary randomized response \cite{warner1965randomized} that was analyzed in the shuffled model by \cite{CheuSUZZ19}. It also rules out the protocol implied by shuffling RAPPOR \cite{erlingsson2014rappor}, and more generally any protocol obtained by the amplification via shuffling approach of \cite{erlingsson2019amplification, BalleBGN19}. Moreover, in the multi-message shuffled setup, the state-of-the-art real summation protocols of \cite{ghazi2019private,balle_privacy_2019constantIKOS}, which rely on the Split-and-Mix procedure~\cite{ishai2006cryptography}, only give approximate-DP.

A different $\shuffledDP^{\geq 1}$ protocol for binary summation can be obtained by instantiating the recent $\shuffledDP^{\geq 1}$ protocols for computing histograms~\cite{anon-power}, with a domain size of $B = 2$. On a high-level, the two resulting protocols---one of which is based on the Count Min sketch and the other on the Hadamard response---can be seen as special cases of the following common template: each user (i) samples a number $\rho$ of messages that depend on their input, (ii) independently samples a number $\eta$ of noise messages, and (iii) sends these $\rho+\eta$ messages to the shuffler. Loosely, the analyzer then outputs the number of messages ``consistent with'' the queried input. However, it can be seen that any protocol following this template will not be pure-DP, as the supports of the distribution of the count observed at the analyzer can shift by $1$ when a single user input is changed. The crucial insight in our pure protocol for binary summation will be to \emph{correlate} the input-dependent messages and the noise messages sampled by each user in steps (i) and (ii) above. By doing so, we not only aim to ensure that the supports are identical but that the two densities are also within a small multiplicative factor on any point. We implement this idea using binary messages by having each user send $d$ bits on both inputs $0$ and $1$. Specifically, the user will start by flipping a suitably biased coin.  If it lands as head, the user will send $(d+1)/2$ zeros and $(d-1)/2$ ones when the input is $0$, and vice versa when the input is $1$. If the coin lands as tail, the user will sample an integer $z$ from a truncated discrete Laplace distribution and send $z$ zeros and $d-z$ ones (see Algorithm~\ref{alg:randomizer} and Equation~(\ref{eq:truncated_discrete_Laplace}) for more details). The overall (mixture) distributions of transmitted ones under both zero and one inputs are superimposed in Figure~\ref{fig:protocol_dist} (in log scale). The analyzer (Algorithm~\ref{alg:analyzer}) then outputs the number of received ones after debiasing. Note that the number of ones received by the analyzer is a random variable taking values between $0$ and $d n$ inclusive. To prove that the algorithm is private, we intuitively wish to argue that the noise distribution satisfies the property that its density values on any two adjacent points are within a multiplicative $e^{\epsilon}$ factor. However, the technical challenge stems from the fact that this noise distribution depends on the specific input sequence (and as we discussed above this dependence is necessary!). Instead, we have to analyze the $n$-fold convolution of the individual responses, and show that the density values of the resulting distribution on any two adjacent points in $\{0,1,\dots, dn\}$ are within a multiplicative factor of $e^{\epsilon}$, for any input sequence.  The crux of the proof is to relate the tails of different convolutions of the truncated discrete Laplace distribution (Lemmas~\ref{lem:tail-bound} and~\ref{lem:pia-lb}).  We determine a setting of (i) the mixture probability coefficient (denoted by $p$ in Algorithm~\ref{alg:randomizer}), (ii) the parameter $d$, and (iii) the ``inverse scaling coefficient'' of the truncated discrete Laplace distribution (denoted by $s$ in Algorithm~\ref{alg:randomizer}), for which the privacy property holds and for which the resulting expected absolute error is $O_{\epsilon}(1)$.

We point out that the dependence of the error on $\epsilon$ that we obtain is $\tilde{O}(1/\epsilon^{3/2})$ for $\epsilon \leq O(1)$ (see Theorem~\ref{th:pure_bin_agg_protocol} for more details). An interesting open question is whether this dependence can be further reduced to $O(1/\epsilon)$, which is the tight error in the central model \cite{dwork2006calibrating}.


\begin{algorithm}[t]
\caption{Randomizer for binary summation.}\label{alg:randomizer}
\begin{algorithmic}[1]
\Procedure{BinaryRandomizer$_{\eps, n}(x)$}{}
\State Let $p, d, s$ be as in Lemma~\ref{lem:main-ratio} (depending on $\eps, n$)
\State $a \leftarrow \mathrm{Ber}(p)$
\If{$a = 0$}
\If{$x = 0$}
\State \Return the multiset with $\left(\frac{d - 1}{2}\right)$ ones and $\left(\frac{d + 1}{2}\right)$ zeros
\Else
\State \Return the multiset with $\left(\frac{d + 1}{2}\right)$ ones and $\left(\frac{d - 1}{2}\right)$ zeros
\EndIf
\Else
\State $z \leftarrow \DLap_d(d/2, s)$
\State \Return the multiset with $z$ ones and $(d - z)$ zeros
\EndIf
\EndProcedure
\end{algorithmic}
\end{algorithm}

\begin{algorithm}[t]
\caption{Analyzer for binary summation.}\label{alg:analyzer}
\begin{algorithmic}[1]
\Procedure{BinaryAnalyzer$_{\eps, n}$}{$R$} 
\State Let $d$ be as in Lemma~\ref{lem:main-ratio} (depending on $\eps, n$)
\State \Return $\frac{n}{2} + \sum_{y\in R} \left(y - \frac{1}{2}\right)$ 
\EndProcedure
\end{algorithmic}
\end{algorithm}

\paragraph{Real Summation Protocol.}

We use our pure private binary summation protocol outlined above as a building block in order to obtain a pure private real summation protocol and prove Theorem~\ref{th:pure_real_agg_protocol_informal}. We note that Cheu et al.~\cite{CheuSUZZ19} had given a transformation from binary summation to real summation, but their reduction results in a protocol with a very large communication of $\tilde{\Omega}(\sqrt{n})$ bits in order to achieve logarithmic error. We instead give a (different) transformation that results in a protocol with polylogarithmic communication. The high-level idea of our reduction is the following: consider the binary representation of the inputs after rounding them to $O(\log n)$ bits of precision, then approximate the sum for each bit position independently, and finally combine the estimates into an approximation of the (real-valued) sum of the inputs.
Since the bit sum estimates have geometrically decreasing weights, we can afford to increase the error on less significant bits. 
In terms of privacy, this means that for the $j$th most significant bit, we run an $\eps_j$-DP binary summation protocol where $\eps_1, \eps_2, \ldots$ is a decreasing sequence. The protocol is illustrated in Algorithms~\ref{alg:real-randomizer} and~\ref{alg:real-analyzer}. By carefully choosing the sequence $\eps_1, \eps_2, \dots $, we can ensure that the total pure privacy parameter $\sum_j \eps_j$ is small, while the total error is a constant times the error for the sum of the most significant bits of the inputs. Intuitively, choosing $\eps_1, \eps_2, \dots$ to be a geometrically decreasing sequence (e.g., $\eps_j = \frac{0.9^j \cdot \eps}{10}$) should suffice for our purposes. However since the communication complexity of our binary summation protocol also depends on the privacy parameter $\eps$, such a choice of the sequence would result in $\poly(n)$ communication complexity. To overcome this, our actual sequence has a ``cut-off'' so that the $\eps_j$'s do not go below a certain value. Please see Section~\ref{sec:real-summation} for more details.

\begin{algorithm}[t]
\caption{Randomizer for real summation.}\label{alg:real-randomizer}
\begin{algorithmic}[1]
\Procedure{RealRandomizer$_{(\eps_j)_{j \in \N}, n}(x)$}{}
\For{$j = 1$ \textbf{to} $2\log n$}
\State $x[j] \leftarrow j$th most significant bit of $x$
\State $S_j \leftarrow${\sc BinaryRandomizer}$_{\eps_j, n}(x[j])$ \hfill \text{$S_j$ is a multiset of zeros and ones.}
\State $R_j \leftarrow \{j\} \times S_j$ \hfill \text{$R_j$ is a multiset of tuples $(j, 0)$ and $(j, 1)$.}
\EndFor
\State \Return $\bigcup_{j=1}^{2\log n} R_j$
\EndProcedure
\end{algorithmic}
\end{algorithm}

\begin{algorithm}[t]
\caption{Analyzer for real summation.}\label{alg:real-analyzer}
\begin{algorithmic}[1]
\Procedure{RealAnalyzer$_{(\eps_j)_{j \in \N}, n}(R)$}{}
\For{$j = 1$ \textbf{to} $2\log_2 n$}
\State $R_j \leftarrow \{y_1 \; |  \; y \in R \text{ and } y_0 = j\}$ \hfill \text{Multiset of bit messages for the $j$th bit.}
\State $a_j \leftarrow ${\sc BinaryAnalyzer}$_{\eps_j, n}(R_j)$
\EndFor
\State \Return $\sum_{j=1}^{2\log n} a_j / 2^j$
\EndProcedure 
\end{algorithmic}
\end{algorithm}

\paragraph{Lower Bound.} 
We next outline the proof of Theorem~\ref{th:pure_bin_agg_lb}. Without loss of generality, we consider an arbitrary $\epsilon$-$\shuffledDP$ protocol performing binary summation with error $n^{0.5-\Omega(1)}$, and where every user sends $m$ messages each belonging to the domain $\{1,\dots,k\}$. We wish to lower bound the number of bits of communication per user in this protocol, which is equal to $m \log{k}$. We denote by $\bX^0$ and $\bX^1$ the random multisets of messages sent by a user in this protocol under inputs $0$ and $1$ respectively. Note that $\bX^0$ and $\bX^1$ are supported on the set $\Delta_{k,m} := \{ (z_1, \ldots, z_k) \in \Z_{\geq 0}^k ~\mid~ z_1 + \cdots + z_k = m\}$. Here, $z_i$ captures the number of $i$ messages sent by the user for each $i \in \{1,\dots,k\}$.

Using the pure privacy of the protocol, we can argue that the ratio of the moment generating functions (MGFs) of $\bX^0$ and $\bX^1$ cannot take a very large or a very small value. Specifically, using the fact that the MGF of a sum of independent random variables is equal to the product of the individual MGFs, we derive a simple yet powerful property that should be satisfied by any $\epsilon$-DP protocol in the shuffled model: the ratio of the MGFs of $\bX^0$ and $\bX^1$ should always lie in the interval $[e^{-\eps}, e^{\eps}]$. We will refer to such random variables as having an \emph{$e^{\eps}$-bounded MGF ratio} (see Section~\ref{sec:mgf} for more details). We remark that while MGFs have been used before in DP by Abadi et al.~\cite{DPSGD} and subsequent works on Renyi DP (starting from~\cite{Renyi-DP}), these usages are in a completely different context compared to ours. In particular, these prior works keep track of the moments in order to bound the privacy parameters under composition of protocols.  To the best of our knowledge, MGFs have neither been used in lower bounds for DP nor in the shuffled model before.

Then, using the accuracy of the protocol, we can deduce that the total variation distance between $\bX^0$ and $\bX^1$ has to be large. We do so by invoking a result from the literature~\cite{ChanSS12,anon-power} showing that for any binary summation protocol that incurs an absolute error of $\alpha$, the total variation distance between $\bX^0$ and $\bX^1$ must be at least $1 -\Theta(\alpha/\sqrt{n})$ (see Theorem~\ref{thm:chan-local} for more details). Since $\alpha = n^{0.5 - \Omega(1)}$ in our case, we get a lower bound of $1 - n^{-\Omega(1)}$ on the total variation distance between $\bX^0$ and $\bX^1$.

Equipped with these two ingredients, the task of lower bounding the per user communication cost of the protocol reduces to lower bounding the following quantity:

\begin{definition} \label{def:lower-bound-question}
Given parameters $\epsilon > 0$ and $\gamma \in [0,1]$, we define $C_{\epsilon, \gamma}$ as the minimum value of $m \log k$ for which there exist two random variables supported on $\Delta_{k, m}$ that are at total variation distance is at least $1 - \gamma$ but that have an $e^{\eps}$-bounded MGF ratio.
\end{definition}



Note that any lower bound on the value of $C_{\epsilon, \gamma}$ can be used to infer a lower bound on the per user communication cost. In order to prove Theorem~\ref{th:pure_bin_agg_lb}, and given our setting of $\gamma = 1/n^{\Omega(1)}$, it is thus enough for us to show that $C_{\epsilon, \gamma} \geq \Omega_{\eps}(\sqrt{\log (1/\gamma)})$. To prove this bound, it suffices to show that if two random variables $\bX^0, \bX^1$ have an $e^{\eps}$-bounded MGF ratio, then their total variation distance must be at least $1 - \exp(O_{\eps}(m^2 \log k))$. For each $\bx \in \Delta_{k, m}$, we view $\Pr[\bX^0 = \bx]$ and $\Pr[\bX^1 = \bx]$ as variables. The $e^\eps$-bounded MGF ratio constraints can then be written as infinitely many linear inequalities over these variables. Moreover, the total variation distance between $\bX^0$ and $\bX^1$ can be written as a maximum of linear combinations of these same variables. We therefore get a linear program with infinitely many constraints, and we would like to show that any solution to it has ``cost'' (i.e., total variation distance) at least $1 - \exp(O_{\eps}(m^2 \log k))$. We do so by giving a dual solution with cost at most $1 - \exp(O_{\eps}(m^2 \log k))$, which by weak duality implies our desired bound (see Section~\ref{sec:lb_bit_agg} for more details).

A natural question is if the lower bound $C_{\epsilon, \gamma} \geq \Omega_{\eps}(\sqrt{\log(1/\gamma)})$ outlined above can be improved, as that would immediately lead to an improved communication complexity lower bound. However, we show that the lower bound is tight, even in the special case where $k = 2$. Namely, we give two random variables supported on $\Delta_{2, m}$ with $m = \Theta_{\eps}(\sqrt{\log(1/\gamma)})$ that are at total variation distance at least $1 - \gamma$ but that have an $e^{\eps}$-bounded MGF ratio. Our construction is based on truncations of discrete Gaussian random variables (see Section~\ref{sec:MGF_limit} for more details).  We note that this limitation only applies to the approach of lower bounding the per user communication complexity via lower bounding $C_{\epsilon, \gamma}$. It remains possible that other approaches might give better lower bounds. For instance, one might be able to proceed by giving a necessary condition for the accuracy of binary summation protocols that is stronger than the total variation distance bound that we used, or a necessary condition for pure privacy that is better than our $e^{\eps}$-bounded MGF ratio property.


\subsection{Organization}
We start with some notation and background in Section~\ref{sec:prelim}. Our protocol for binary summation is presented and analyzed in Section~\ref{sec:bit-sum-protocol}. In Section~\ref{sec:lb_bit_agg}, we prove our lower bound (Theorem~\ref{th:pure_bin_agg_lb}). Our protocol for real summation appears in Section~\ref{sec:real-summation}. We conclude with some interesting open questions in Section~\ref{sec:conc_oqs}. Our corollary for histograms appears in Appendix~\ref{sec:app_hist}, and deferred proofs appear in Appendices~\ref{sec:missing_proofs_bit_sum_prot},~\ref{app:tv-vs-utility-proof}, and~\ref{app:discrete-gaussian-tail}.

\section{Preliminaries}\label{sec:prelim}

\paragraph{Shuffled Model of Privacy.}

We denote by $n$ the number of users. For each $i$ in $[n] := \{1,\dots,n\}$, we denote by $x_i$ the input held by the $i$th user, and further assume that $x_i \in \MX$.  In the binary summation case, we have that $\MX = \{ 0, 1\}$ while in the real summation case, we let $\MX$ be the set $[0, 1]$ of real numbers.  A protocol $P = (R, S, A)$ in the shuffled model consists of three algorithms: (i) the {\it local randomizer} $R(\cdot)$ whose input is the data of one user and whose output is a sequence of messages, (ii) the {\it shuffler} $S(\cdot)$ whose input is the concatenation of the outputs of the local randomizers and whose output is a uniform random permutation of its inputs, and (iii) the {\it analyzer} $A(\cdot)$ whose input is the output of the shuffler and whose output is the output of the protocol. An illustration of the shuffled model is given in Figure~\ref{fig:esa}. The privacy in the shuffled model is guaranteed with respect to the input to the analyzer, i.e., the output of the shuffler.  
\begin{definition}[DP in the shuffled model,~\cite{erlingsson2019amplification, CheuSUZZ19}]
\label{def:dp_shuffled}
 A protocol $P = (R, S, A)$ is {\it $(\ep, \delta)$-$\shuffledDP$} if, for any dataset $X = (x_1, \ldots, x_n)$, the algorithm 
$S(R(x_1), \ldots, R(x_n))$
is $(\ep, \delta)$-DP. In the special case where $\delta = 0$, we say that the protocol $P$ is $\epsilon$-$\shuffledDP$.
\end{definition}
Note that the $\localDP$ model corresponds to the case where $S$ is replaced by the identity function.  

\begin{definition}[Non-Interactive Protocols] \label{def:non-interactive-protocols}
Let $k$ and $m$ be positive integers, and let $\Delta_{k,m} := \{ (z_1, \ldots, z_k) \in \Z_{\geq 0}^k ~\mid~ z_1 + \cdots + z_k = m\}$.  In a \emph{non-interactive} (aka \emph{one-round}) protocol, each of the $n$ users (i.e., randomizers) receives an input $b$ and outputs at most $m$ messages each consisting of $\log k$ bits, according to a certain distribution (depending on $b$), and using private randomness. We say that such a protocol has a \emph{communication complexity} of $m \log k$.
\end{definition}
It is often convenient to view each message as a number in $[k]$. We use $\bX^b \in \Z_{\geq 0}^{k}$ to denote the random variable whose $s$th coordinate $X^b_{s}$ denotes the number of $s$-messages output by the randomizer on input $b$.  Note that it is always the case that $\sum_{s \in [k]} X^b_s = m$, i.e., $\supp(\bX^b) \subseteq \Delta_{k, m}$.

\begin{figure}[h]
\centering
\includegraphics[width=0.73\textwidth]{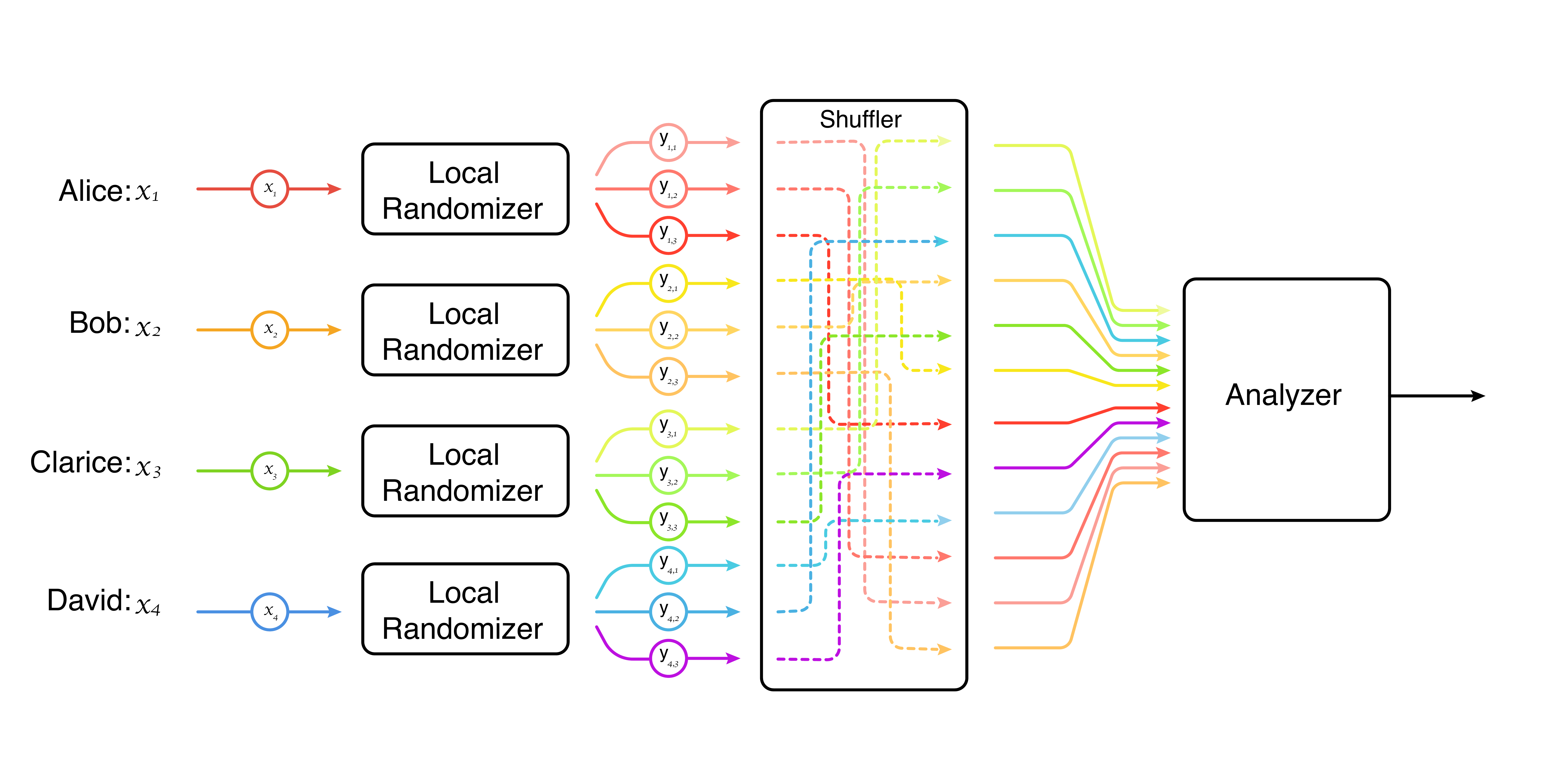}
\caption{In the shuffled model, the inputs are first locally randomized, yielding a number of messages that are sent to the shuffler. The shuffler then randomly permutes all incoming messages before passing them to the analyzer. This figure is reproduced from~\cite{anon-power}.}
\label{fig:esa}
\end{figure}

\begin{figure}[h]
    \centering
    \includegraphics[width=0.5\textwidth]{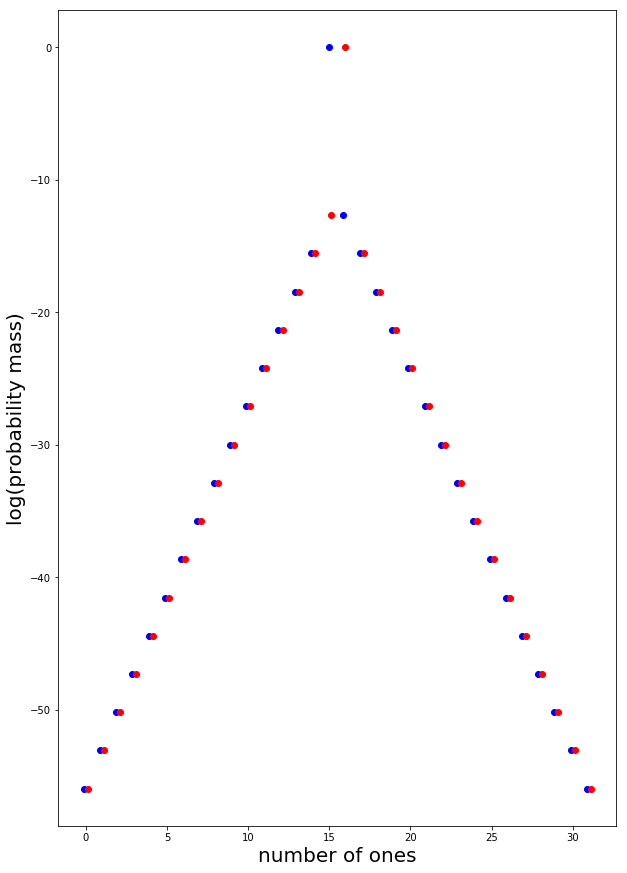}
    \caption{An illustration of the probability mass functions of the number of ones output by the randomizer (Algorithm~\ref{alg:randomizer}) for parameter $d = 31, s = 0.5, p = 0.01$. The $x$-axis corresponds to the number of ones and the $y$-axis corresponds to the base-$2$ logarithm of the probability. The red points and the blue points correspond to when the input is one and zero respectively.
    }
    \label{fig:protocol_dist}
\end{figure}

\section{Pure Binary Summation Protocol via Shuffling}
\label{sec:bit-sum-protocol}
In this section we prove Theorem \ref{th:pure_bin_agg_protocol_informal}, restated formally below.
\begin{theorem}\label{th:pure_bin_agg_protocol}
	For every sufficiently large $n$ and $O(1) \geq \eps > 1/n^{2/3}$, there is an $\eps$-$\shuffledDP$ protocol for summation for inputs $x_1, \ldots, x_n\in\{0,1\}$ where each user sends $O\left(\frac{\log n}{\ep}\right)$ one-bit messages to the analyzer and has expected error at most $O\left( \frac{\sqrt{ \log 1/\ep}}{\ep^{3/2}}\right)$. 
\end{theorem}

We remark that the assumption that $\varepsilon > \frac{1}{n^{2/3}}$ is made without loss of generality, because, for $\varepsilon \leq \frac{1}{{n}^{2/3}}$, there is a trivial algorithm that achieves square error of $O(1/\eps^{3/2})$: the analyzer just always outputs 0.

Throughout this section we assume that for some absolute constant $C$, $\ep \leq C$, and thus in particular $e^{\ep}$ can be bounded above by an absolute constant. (The constant $C$ can be arbitrary.) It is well-known that any $\ep$-$\centralDP$ protocol for summation has error $\Omega(1/\ep)$ \cite{Vadhan-tutorial}. 
Thus the error in Theorem \ref{th:pure_bin_agg_protocol} is suboptimal by a factor of at most $\tilde{O}(1/\sqrt{\ep})$. 

The remainder of the section is organized as follows. In Section \ref{sec:protocol}, we present the protocol used to prove Theorem \ref{th:pure_bin_agg_protocol}. In Section \ref{sec:protocol_analysis}, we prove the accuracy and privacy guarantees of Theorem \ref{th:pure_bin_agg_protocol}, and in Section \ref{sec:tail-bound} we prove a technical lemma needed in the privacy analysis.

\subsection{The Protocol}
\label{sec:protocol}

To described the protocol, let us recall the discrete Laplace (aka symmetric Geometric) distribution. For notational convenience, we identify the discrete Laplace distribution by two parameters: the mean $\mu$ and the ``inverse scaling exponent'' $s > 1$. The discrete Laplace distribution associated with these parameters, denoted by $\DLap(\mu, s)$, has the following probability mass function: for $z \in \BZ$,
\begin{align*}
\Pr_{Z \sim \DLap(\mu, s)}[Z = z] = \frac{1}{C(\mu, s)} \cdot e^{- |z - \mu| / s},
\end{align*}
where $C(\mu, s) = \sum_{z = -\infty}^{\infty} e^{-|z - \mu| / s}$ is the normalization factor.

We will use the truncated version of the discrete Laplace distribution, for which we condition the support to be on $[\mu - w/2, \mu + w/2]$ where $w \geq 1$ is the ``width'' of the support. We denote such a distribution by $\DLap_w(\mu, s)$. In other words, its probability mass function satisfies
\begin{align}
\Pr_{Z \sim \DLap_w(\mu, s)}[Z = z] = 
\begin{cases}
\frac{1}{C_w(\mu, s)} \cdot e^{- |z - \mu| / s} & \text{if } \mu - w/2 \leq z \leq \mu + w/2 \text{ and } z \in \BZ\\
0 & \text{otherwise.}
\end{cases}\label{eq:truncated_discrete_Laplace}
\end{align}
Once again $C_w(\mu, s) = \sum_{z \in [\mu - w/2, \mu + w/2] \cap \Z} e^{- |z - \mu| / s}$ is simply the normalization factor.

Our randomizer and analyzer are presented in Algorithm~\ref{alg:randomizer} and Algorithm~\ref{alg:analyzer}, respectively. 
The protocol has 3 parameters: the number of messages $d$, the  ``inverse scaling exponent'' $s$, and the ``noise  probability'' $p$. We always assume that $d$ is a positive odd integer\footnote{We only assume that $d$ is odd for convenience, so that $\left(\frac{d - 1}{2}\right)$ and $\left(\frac{d + 1}{2}\right)$ are integers. Using an even $d$ and replacing these two quantities with $d/2 - 1$ and $d/2 + 1$ also works, provided that the proofs are adjusted appropriately.}.
These parameters will be chosen later (in Lemma~\ref{lem:main-ratio}).

\subsection{Privacy Analysis}
\label{sec:protocol_analysis}

For $b \in \{0, 1\}$, we write $\cR_b$ to denote the distribution\footnote{This is the distribution of $\bX^b$ defined in Section~\ref{sec:prelim}.} on the number of ones output by the randomizer on input $b$. (This distribution depends on $d, s, p$ but we do not include them in the notation to avoid being cumbersome.) Notice that we can decompose $\cR_b$ as a mixture $p \cdot \DLap_d(d/2, s) + (1 - p) \cdot \bone(\frac{d - 1}{2} + b)$, where we use $\bone(\vartheta)$ to denote the distribution that is $\vartheta$ with probability 1.

To prove the privacy guarantee of Theorem \ref{th:pure_bin_agg_protocol}, we first note that we may focus only on the neighboring datasets $(0, \ldots, 0, 0)$ and $(0, \ldots, 0, 1)$; this follows since we may assume (due to symmetry) that more than half of the bits are zero and we can then condition out the results from the 1 bits that they share. (See the proof of Theorem~\ref{th:pure_bin_agg_protocol} for a formalization of this.) For these datasets, Lemma \ref{lem:main-ratio} below bounds the ratio of the probabilities of ending up with a particular union of outputs from these two datasets.

\newcommand{\scalingParameter}[1]{\frac{10}{{#1}}}
\newcommand{\noiseProbabilityBig}[2]{\frac{100\,e^{100{#1}}}{{#2}(1-e^{-0.1{#1}})^2}}
\newcommand{\noiseProbabilitySmall}[2]{\frac{100\,e^{100{#1}}\log(1/(1-e^{-0.1{#1}}))}{{#2}(1-e^{-0.1{#1}})}}
\newcommand{\noiseProbabilityLB}[2]{\frac{100\,e^{100{#1}}}{{#2}(1-e^{-0.1{#1}})}}

\newcommand{\numberOfMessagesSmall}[2]{4 \left\lceil\frac{1000\, e^{100{#1}}}{(1 - e^{-0.1{#1}})} \cdot \log \left(\frac{{#2}}{1-e^{-0.1{#1}}} \right)\right\rceil + 3}
\newcommand{\numberOfMessagesBig}[2]{4 \left\lceil\frac{1000\, e^{100{#1}}}{(1 - e^{-0.1{#1}})^2} \cdot \log \left(\frac{{#2}}{1-e^{-0.1{#1}}} \right)\right\rceil + 3}


\begin{lemma} \label{lem:main-ratio}
There is a sufficiently small constant $c_0 \in (0,1)$ so that the following holds. For any sufficiently large $n \in \N$ and any $c_0 \geq \varepsilon > \frac{1}{{n}^{2/3}}$, let $s = \scalingParameter{\eps}$, $p = \noiseProbabilitySmall{\ep}{n}$, 
  and $d = \numberOfMessagesSmall{\eps}{n}$. 
  Then, we have 
\begin{align} \label{eq:dp-ratio}
\frac{\Pr_{Z_1, \dots, Z_n \sim \cR_0}[Z_1 + \cdots + Z_n = t]}{\Pr_{Z_1, \dots, Z_{n - 1} \sim \cR_0, Z_n \sim \cR_1}[Z_1 + \cdots + Z_n = t]} \in [e^{-\varepsilon}, e^{\varepsilon}],
\end{align}
for all $t \in \{0, \dots, dn\}$.
\end{lemma}

This means that, for the above selection of parameters, the protocol is $\varepsilon$-DP. Using Lemma \ref{lem:main-ratio}, we prove Theorem \ref{th:pure_bin_agg_protocol}.
\begin{proof}[Proof of Theorem \ref{th:pure_bin_agg_protocol}]
We may assume without loss of generality that $\eps \leq c_0$, as otherwise we may set $\eps$ to $\min\{\eps, c_0\}$ instead.

  We use the local randomizer {\sc BinaryRandomizer$_{\eps, n}$} of Algorithm \ref{alg:randomizer} and the analyzer {\sc BinaryAnalyzer$_{\eps, n}$} of Algorithm \ref{alg:analyzer}, with the parameters $s,d,p$ given by the expressions in Lemma \ref{lem:main-ratio}, except with $n$ replaced by $\lceil (n+1)/2 \rceil$. Explicitly, we have $s = \scalingParameter{\ep}$, $p = \noiseProbabilitySmall{\ep}{\lceil (n+1)/2 \rceil}$ and $d = \numberOfMessagesSmall{\eps}{\lceil (n+1)/2 \rceil}$.
  We prove the accuracy guarantee first, which is a simple consequence of the choices of $p,d$ made in Lemma \ref{lem:main-ratio}, followed by the privacy guarantee, which uses Lemma \ref{lem:main-ratio}.


  \paragraph{Proof of accuracy.}
  Fix a dataset $X = (x_1, \ldots, x_n) \in \{0,1\}^n$. Let $Y \in \R$ be the count released by the analyzer. Moreover, for $1 \leq i \leq n$, let $Z_1, \ldots, Z_n$ be i.i.d. random variables distributed according to $\nu = \DLap_d(d/2, s)$. It is easy to check that $\Var[Z_i] \leq 2s^2$. 
  Moreover, let $M \in \{1, 2, \ldots, n\}$ be the number of users for whom the Bernoulli random variable $a$ is equal to 1. 
  In particular, $M \sim \Binom(n,p)$. 
 The expected absolute error is given by
  \begin{align}
    \E\left[ \left|Y - \sum_{i=1}^n x_i \right| \right]
&\leq 
      \sum_{m=0}^n\Pr[M = m] \cdot \left( \frac{m}{2} + \E\left[\left|Z_1 + \cdots + Z_m - \frac{md}{2} \right| \right]\right) \nonumber\\
    \text{(by Jensen's inequality) }     & \leq \E[M/2] + \sum_{m=0}^n \Pr[M=m] \cdot \sqrt{\E \left[ \left(Z_1 + \cdots + Z_m - \frac{md}{2} \right)^2 \right]}\nonumber\\
    & \leq pn/2 + \sum_{m=0}^n \Pr[M=m] \cdot \sqrt{ \E \left[ \sum_{j=1}^m (Z_i - d/2)^2 \right]} \nonumber\\
  \text{(since } Z_1, \ldots, Z_n \text{ are iid) }  & = pn/2 + \sum_{m=0}^n \Pr[M=m] \cdot \sqrt{m} \cdot \sqrt{\Var[Z_1]} \nonumber\\
   & \leq pn/2 + \sum_{m=0}^n \Pr[M=m] \cdot \sqrt{m} \cdot \sqrt{2s^2}\nonumber\\
    & = pn/2 + \sqrt{2} s \cdot \E[\sqrt{M}] \nonumber\\
    \label{eq:abs-error}
  & \leq pn/2 + \sqrt{2} s \cdot \sqrt{pn}.
  \end{align}
  Since $p = \noiseProbabilitySmall{\ep}{\lceil (n+1)/2 \rceil}$, we have $pn/2 + \sqrt{2pn} \cdot s = O\left(\frac{\sqrt{\log 1/\ep}}{\ep^{3/2}} \right)$. Combined with (\ref{eq:abs-error}), this gives us the desired upper bound on the expected error of the protocol.

  \paragraph{Proof of privacy.} Let $X = (x_1, \ldots, x_{n-1}, x_n) \in \{0,1\}^n$ and $X' = (x_1, \ldots, x_{n-1}, x_n') \in \{0,1\}^n$ be two neighboring datasets. By symmetry, without loss of generality, we may assume that $x_n = 0$ and at least $n_0 := \lceil n-1\rceil/2$ of the values of $x_1, \ldots, x_{n-1}$ are also 0. By permuting the users, we may also assume without loss of generality that $x_1 = x_2 = \cdots = x_{n_0} = 0$. For $1 \leq i \leq n$, let $Y_i \in [0,d]$ denote the (random) number of 1s output by user $i$ when their input is $x_i$. Also let $Y_n' \in [0,1]$ denote the (random) number of 1's output by user $n$ when its input is $x_n'$.  By \cite[Lemma A.2]{BalleBGN19}, to show that for all $t \in \BN$,
  $$
  \frac{\Pr[Y_1 + \cdots + Y_{n-1} + Y_n = t]}{\Pr[Y_1 + \cdots + Y_{n-1} + Y_n' = t]} \in [e^{-\ep}, e^{\ep}],
  $$
  it suffices to show that for all $t_0 \in \BN$,
  \begin{equation}
  \label{eq:t0}
  \frac{\Pr[Y_1 + \cdots +  Y_{n_0} + Y_n = t_0]}{\Pr[Y_1 + \cdots + Y_{n_0} + Y_n' = t_0]} \in [e^{-\ep}, e^{\ep}].
  \end{equation}
  Now the validity of (\ref{eq:t0}) is an immediate consequence of Lemma \ref{lem:main-ratio} with the parameter $n$ of Lemma \ref{lem:main-ratio} equal to $n_0 + 1$.
\end{proof}


From now on, we will use $\nu$ and $\omega_b$ as abbreviations for $\DLap_d(d/2, s)$ and $\bone(\frac{d - 1}{2} + b)$ respectively, where $d, s$ are defined as in Lemma~\ref{lem:main-ratio}. 

Let us denote by $P_{m, k}$ the probability that $m$ independent random variables from the noise distribution $\nu$ sums up to $k$; more formally,
\begin{align*}
P_{m, k} := \Pr_{Z_1, \dots, Z_m \sim \nu}[Z_1 + \cdots + Z_m = k].
\end{align*}
For convenience, we define $P_{0, 0} = 1$ and $P_{0, k} = 0$ for all $k \ne 0$.

As we will see in the proof of Lemma~\ref{lem:main-ratio} below, expansions of the numerator and denominator of the left hand side of~\eqref{eq:dp-ratio} result in similar terms involving $P_{m, k}$, except occasionally with (i) $k$ differing by one or (ii) $m$ differing by 1 and $k$ differing by $\left(\frac{d - 1}{2}\right)$ or $\left(\frac{d - 3}{2}\right)$. Hence, to bound the ratio between the two, we have to find some relation between $P_{m, k}, P_{m, k - 1}, P_{m + 1, k + \left(\frac{d - 1}{2}\right)}$, and $P_{m + 1, k + \left(\frac{d - 3}{2}\right)}$. The exact inequality we will use here is stated below and proved in Section~\ref{sec:tail-bound}.

\newcommand{\mLowerBound}[1]{\frac{10\log(1/(1-e^{-0.1{#1}}))}{1-e^{-0.1{#1}}}}

\begin{lemma} \label{lem:tail-bound}
  For any sufficiently large $n \in \N$, let $\eps, d$ and $s$ be as in Lemma~\ref{lem:main-ratio}. Then the following hold:
  For any integers $\mLowerBound{\ep} \leq m \leq n - 1$, and $\ell_1, \ell_2 \in \left\{\frac{d - 1}{2}, \frac{d - 3}{2}\right\}$, if $p \geq \noiseProbabilityLB{\ep}{n}$, then we have
\begin{align}
  \label{eq:tail-bound-lemma}
e^{-\ep} p(1 - e^{-\eps/2}) \cdot \left(P_{m+1, k+\ell_1} + \frac{(n - m - 1)}{m + 1} \cdot P_{m + 1, k+\ell_2}\right) + e^{0.2\ep} \cdot P_{m, k - 1} \geq P_{m, k}.
\end{align}
\end{lemma}

We prove Lemma \ref{lem:tail-bound} in Section \ref{sec:final-tail-bound}. 
We additionally need the following Lemma \ref{lem:pia-lb}, which can be interpreted as a sort of anti-concentration result. Recall that $P_{i,j} = \Pr[Z_1 + \cdots + Z_i = j]$, where $Z_1, \ldots, Z_i \sim \nu = \DLap_d(d/2, s)$. For any $a \in \BN$, if also $Z_1', \ldots, Z_a' \sim \nu$, then as $\E[Z_1' + \cdots + Z_a'] = da/2$ and the distribution of $Z_1' + \cdots + Z_a'$ has sufficient mass at its expectation, one should expect that $P_{i+a,j+a} = \Pr[Z_1 + \cdots + Z_i + Z_1' + \cdots + Z_a' = j + da/2]$ is not too much smaller than $P_{i,j}$. Lemma \ref{lem:pia-lb} says that in fact $\Pr[Z_1 + \cdots + Z_i + Z_1' + \cdots + Z_a' = j + da/2 - d/2]$ is not too much smaller than $P_{i,j}$.
\begin{lemma}
  \label{lem:pia-lb}
  For any $i, j, a \in \BN_0$ such that $a \leq s^2/1000$, we have
  $$
P_{i+a, j + a \left( \frac{d-1}{2} \right)} \geq  \frac{\sqrt{a}}{40s^3}  \cdot P_{i,j}.
  $$
\end{lemma}
The proof of Lemma \ref{lem:pia-lb} is deferred to Section \ref{sec:pia-lb}. We note that the multiplicative factor on the right hand side of the above lemma is unimportant; in fact, as long as it is $1/s^{O(1)}$, it suffices for our proof.

With Lemmas~\ref{lem:tail-bound} and \ref{lem:pia-lb} ready, we can now prove Lemma~\ref{lem:main-ratio} as follows.

\begin{proof}[Proof of Lemma~\ref{lem:main-ratio}]
Let $c_0 \in (0,1)$ be some sufficiently small positive constant, to be specified later. We would like to show that, for all $t \in \{0, \dots, dn\}$, the following two inequalities hold: 
\begin{align} \label{eq:target-ratio-first}
\Pr_{Z_1, \dots, Z_n \sim \cR_0}[Z_1 + \cdots + Z_n = t] \leq e^{\varepsilon} \cdot \Pr_{Z_1, \dots, Z_{n - 1} \sim \cR_0, Z_n \sim \cR_1}[Z_1 + \cdots + Z_n = t],
\end{align}
and
\begin{align} \label{eq:target-ratio-second}
\Pr_{Z_1, \dots, Z_{n - 1} \sim \cR_0, Z_n \sim \cR_1}[Z_1 + \cdots + Z_n = t] \leq e^{\varepsilon} \cdot \Pr_{Z_1, \dots, Z_n \sim \cR_0}[Z_1 + \cdots + Z_n = t].
\end{align}

\paragraph{Proof of (\ref{eq:target-ratio-first}).} We will start by showing~\eqref{eq:target-ratio-first}. To do so, let us first decompose the probability on the left and the right hand sides based on whether $Z_n$ is sampled from the noise distribution $\nu$.
This gives
\begin{align}
&\Pr_{Z_1, \dots, Z_n \sim \cR_0}[Z_1 + \cdots + Z_n = t] \nonumber \\
  &= p \cdot \Pr_{Z_1, \dots, Z_{n - 1} \sim \cR_0, Z_n \sim \nu}[Z_1 + \cdots + Z_n = t] \nonumber \\
  \label{eq:all0}
&\qquad + (1 - p) \cdot \Pr_{Z_1, \dots, Z_{n - 1} \sim \cR_0}\left[Z_1 + \cdots + Z_{n - 1} = t - \left(\frac{d - 1}{2}\right)\right],
\end{align}
and 
\begin{align}
&\Pr_{Z_1, \dots, Z_{n - 1} \sim \cR_0, Z_n \sim \cR_1}[Z_1 + \cdots + Z_n = t] \nonumber \\
  &= p \cdot \Pr_{Z_1, \dots, Z_{n - 1} \sim \cR_0, Z_n \sim \nu}[Z_1 + \cdots + Z_n = t] \nonumber \\
  \label{eq:all0-but1}
&\qquad + (1 - p) \cdot \Pr_{Z_1, \dots, Z_{n - 1} \sim \cR_0}\left[Z_1 + \cdots + Z_{n - 1} = t - \left(\frac{d + 1}{2}\right)\right].
\end{align}
Furthermore, observe that, by expanding based on the number of variables among $Z_1, \dots, Z_{n - 1}$ that uses the noise distribution (i.e., $i$ below), we have
\begin{align}
\Pr_{Z_1, \dots, Z_{n - 1} \sim \cR_0, Z_n \sim \nu}[Z_1 + \cdots + Z_n = t] 
= \sum_{i=0}^{n - 1} \binom{n - 1}{i} p^{i} (1 - p)^{n - 1 - i} \cdot P_{i + 1, t - (n - 1 - i)\left(\frac{d - 1}{2}\right)},
\end{align}
and
\begin{align} \label{eq:expand-plus-one}
\Pr_{Z_1, \dots, Z_{n - 1} \sim \cR_0}\left[Z_1 + \cdots + Z_{n - 1} = t - \left(\frac{d - 1}{2}\right)\right] = \sum_{i=0}^{n - 1} \binom{n - 1}{i} p^{i} (1 - p)^{n - 1 - i} \cdot P_{i, t - (n - i)\left(\frac{d - 1}{2}\right)},
\end{align}
and 
\begin{align} \label{eq:expand-minus-one}
\Pr_{Z_1, \dots, Z_{n - 1} \sim \cR_0}\left[Z_1 + \cdots + Z_{n - 1} = t - \left(\frac{d + 1}{2}\right)\right] = \sum_{i=0}^{n - 1} \binom{n - 1}{i} p^{i} (1 - p)^{n - 1 - i} \cdot P_{i, t - (n - i)\left(\frac{d - 1}{2}\right) - 1}.
\end{align}
We may expand the right hand side of~\eqref{eq:expand-minus-one} further as
\begin{align}
&\sum_{i=0}^{n - 1} \binom{n - 1}{i} p^{i} (1 - p)^{n - 1 - i} \cdot P_{i, t - (n - i)\left(\frac{d - 1}{2}\right) - 1} \nonumber \\
&= \frac{1}{e^{\eps}} \cdot \sum_{i=0}^{n - 1} \binom{n - 1}{i} p^{i} (1 - p)^{n - 1 - i} \cdot  e^{\eps/2} \cdot  P_{i, t - (n - i) \left(\frac{d - 1}{2}\right) - 1} \nonumber \\
&\qquad + \frac{1}{e^{\eps}} \cdot \sum_{i=0}^{n - 1} \binom{n - 1}{i} p^{i} (1 - p)^{n - 1 - i} \cdot (e^{\eps} - e^{\eps/2})P_{i, t - (n - i)\left(\frac{d - 1}{2}\right) - 1} \nonumber \\
&\geq \frac{1}{e^{\eps}} \cdot \sum_{i=0}^{n - 1} \binom{n - 1}{i} p^{i} (1 - p)^{n - 1 - i} \cdot  e^{\eps/2} \cdot  P_{i, t - (n - i) \left(\frac{d - 1}{2}\right) - 1} \nonumber \\
&\qquad + \frac{1}{e^{\eps}} \cdot \sum_{i=0}^{n - 1} \binom{n - 1}{i} p^{i} (1 - p)^{n - 1 - i} \cdot (e^{\eps/2} - 1)P_{i, t - (n - i)\left(\frac{d - 1}{2}\right) - 1} \nonumber \\
&\geq \frac{1}{e^{\eps}} \cdot \sum_{i=0}^{n - 1} \binom{n - 1}{i} p^{i} (1 - p)^{n - 1 - i} \cdot  e^{\eps/2} \cdot  P_{i, t - (n - i) \left(\frac{d - 1}{2}\right) - 1} \nonumber \\
&\qquad + \frac{1}{e^{\eps}} \cdot \sum_{i=0}^{n - 1} \binom{n - 1}{i + 1} p^{i + 1} (1 - p)^{n - 2 - i} \cdot (e^{\eps/2} - 1)P_{i + 1, t - (n - i - 1)\left(\frac{d - 1}{2}\right) - 1} \nonumber \\
&\geq \frac{1}{e^{\eps}} \cdot \sum_{i=0}^{n - 1} \binom{n - 1}{i} p^{i} (1 - p)^{n - 1 - i}  e^{\eps/2} \cdot  P_{i, t - (n - i) \left(\frac{d - 1}{2}\right) - 1} \nonumber \\
&\qquad + \frac{1}{e^{\eps}} \cdot \sum_{i=0}^{n - 1} \binom{n - 1}{i} p^{i} (1 - p)^{n - 1 - i} \cdot \frac{p(n - 1 - i)}{i + 1} \cdot (e^{\eps/2} - 1) P_{i + 1, t - (n - i - 1)\left(\frac{d - 1}{2}\right) - 1}. \label{eq:final-borrow-right-first}
\end{align}

Using the above expressions, we may write the difference between the right hand side and the left hand side of~\eqref{eq:target-ratio-first} as
\begin{align}
&e^{\varepsilon} \cdot \Pr_{Z_1, \dots, Z_{n - 1} \sim \cR_0, Z_n \sim \cR_1}[Z_1 + \cdots + Z_n = t] - \Pr_{Z_1, \dots, Z_n \sim \cR_0}[Z_1 + \cdots + Z_n = t] \nonumber \\
&\geq (e^{\varepsilon} - 1) \cdot p \cdot \sum_{i=0}^{n - 1} \binom{n - 1}{i} p^{i} (1 - p)^{n - 1 - i} \cdot P_{i + 1, t - (n - 1 - i)\left(\frac{d - 1}{2}\right)} \nonumber \\
&\qquad + (1 - p) \cdot \sum_{i=0}^{n - 1} \binom{n - 1}{i} p^{i} (1 - p)^{n - 1 - i}  e^{\eps/2} \cdot  P_{i, t - (n - i) \left(\frac{d - 1}{2}\right) - 1} \nonumber \\
&\qquad + (1 - p) \cdot \sum_{i=0}^{n - 1} \binom{n - 1}{i} p^{i} (1 - p)^{n - 1 - i} \cdot \frac{p(n - 1 - i)}{i + 1} \cdot (e^{\eps/2} - 1) P_{i + 1, t - (n - i - 1)\left(\frac{d - 1}{2}\right) - 1} \nonumber \\
&\qquad - (1 - p) \cdot \sum_{i=0}^{n - 1} \binom{n - 1}{i} p^{i} (1 - p)^{n - 1 - i} \cdot P_{i, t - (n - i)\left(\frac{d - 1}{2}\right)} \nonumber \\
&\geq (1 - p) \sum_{i=0}^{n - 1} \binom{n - 1}{i} p^{i} (1 - p)^{n - 1 - i} \Delta_{i}, \label{eq:sum-delta}
\end{align}
where
\begin{align*}
\Delta_i &:= p(1 - e^{-\eps/2}) \cdot \left(P_{i + 1, t - (n - 1 - i)\left(\frac{d - 1}{2}\right)} + \frac{n - 1 - i}{i + 1} \cdot P_{i + 1, t - (n - i - 1)\left(\frac{d - 1}{2}\right) - 1} \right) \\
& \qquad + e^{\eps/2} \cdot P_{i, t - (n - i) \left(\frac{d - 1}{2}\right) - 1} - P_{i, t - (n - i)\left(\frac{d - 1}{2}\right)},
\end{align*}
and we have used that $e^\ep - 1 \geq e^{\ep/2} - 1 \geq 1- e^{-\ep/2}$ for $\ep \geq 0$. 

 By Lemma \ref{lem:tail-bound} with $m = i, k = t - (n - i)\left(\frac{d - 1}{2}\right), \ell_1 = \left(\frac{d - 1}{2}\right)$, and $\ell_2 = \left(\frac{d - 3}{2}\right)$, we see that
\begin{equation}
  \label{eq:deltai-lb}
  \Delta_i \geq (e^{0.3\ep}-1) P_{i, t-(n-i) \left( \frac{d-1}{2} \right)} \geq 0,
\end{equation}
for all $i$ such that $\mLowerBound{\ep} \leq i \leq n-1$. For ease of notation set $i_0 := \mLowerBound{\ep}$. 
It remains to lower bound the terms in (\ref{eq:sum-delta}) given by $ 0\leq i < i_0$. To do so, we will ``borrow'' the additional mass of $(e^{0.3\ep} - 1) P_{i, t-(n-i) \left( \frac{d-1}{2} \right)}$ from the terms with $i \geq i_0$. To show that this borrowing gives sufficient positive mass from the terms $P_{i, t-(n-i) \left( \frac{d-1}{2} \right)}$ with $i \geq i_0$, we will use Lemma \ref{lem:pia-lb}. 

Next, let $i_{\max} \in \{0, 1, \ldots, i_0 - 1\}$ and  $i_{\min} \in \{ i_0, i_0 + 1, \ldots, 2p(n-1)\}$ be defined so that:
\begin{align*}
  P_{i_{\max}, t-(n-i_{\max}) \left( \frac{d-1}{2} \right)} & \geq P_{i, t-(n-i)\left( \frac{d-1}{2} \right)} \quad \forall i \in \{0, 1, \ldots, i_0-1 \}\\
  P_{i_{\min}, t-(n-i_{\min}) \left( \frac{d-1}{2} \right)} & \leq P_{i, t-(n-i) \left( \frac{d-1}{2} \right)} \quad \forall i \in \{i_0, i_0 + 1, \ldots, 2p(n-1) \}.
\end{align*}
As $p = \noiseProbabilitySmall{\ep}{n}$
and $\ep < c_0 \leq 1$, we have that as long as $c_0$ is sufficiently small,
$$
2p(n-1) \leq \frac{100 e^{100} \log(5/\ep)}{0.1\ep} \leq \frac{1}{4\ep^2} = s^2/1000.
$$
It follows from Lemma \ref{lem:pia-lb} with $a = i_{\min} - i_{\max} \leq 2p(n-1)$ that
$$
P_{i_{\min}, t-(n-i_{\min}) \left( \frac{d-1}{2} \right)} \geq \frac{1}{40s^3} P_{i_{\max}, t-(n-i_{\max}) \left( \frac{d-1}{2} \right)}.
$$

Let $M \sim \Binom(n-1, p)$ be a binomial random variable. Then, as (\ref{eq:deltai-lb}) holds for $n-1 \geq i \geq i_0$, we have
\begin{align}
  & \sum_{i=0}^{n - 1} \binom{n - 1}{i} p^{i} (1 - p)^{n - 1 - i} \Delta_{i}\nonumber\\
  & \geq -\sum_{i=0}^{i_0-1} \binom{n-1}{i} p^i (1-p)^{n-1-i} P_{i, t-(n-i) \left( \frac{d-1}{2} \right)} + \sum_{i = i_0}^{2p(n-1)} \binom{n-1}{i} p^i (1-p)^{n-1-i} (e^{0.3\ep} - 1) \cdot P_{i, t-(n-i) \left( \frac{d-1}{2} \right)} \nonumber\\
  & \geq -P_{i_{\max}, t-(n-i_{\max}) \left(\frac{d-1}{2} \right)} \sum_{i=0}^{i_0 - 1} \binom{n-1}{i} p^i (1-p)^{n-1-i} + 0.3 \ep \cdot P_{i_{\min}, t-(n-i_{\min}) \left( \frac{d-1}{2} \right)} \sum_{i = i_0}^{2p(n-1)} \binom{n-1}{i} p^i (1-p)^{n-1-i} \nonumber\\
  \label{eq:pimax-final}
  & \geq P_{i_{\max}, t-(n-i_{\max}) \left( \frac{d-1}{2} \right)} \cdot \left(\modified{\frac{0.3\ep}{40 s^3}} \cdot \Pr[i_0 \leq M \leq 2p(n-1)] - \Pr[M < i_0 ] \right).
\end{align}

By the Chernoff bound, for sufficiently large $n$ and since $pn = n \cdot \noiseProbabilitySmall{\ep}{n} \geq 100$, we have
$$
\Pr[M > 2p(n-1)] \leq \exp(-p(n-1)/3) \leq \exp(-pn/4) < 1/2.
$$
Moreover, since $i_0 =\mLowerBound{\ep} \leq pn/3 \leq p(n-1)/2$ and $\ep \leq 1$ in the current case,
$$
\Pr[M < i_0] \leq \exp(-p(n-1)/8) \leq \exp(-pn/10) \leq \exp\left( \frac{10}{1-e^{-0.1\ep}} \right) \leq \exp(-1/(2\ep)) < 1/4.
$$
Hence, recalling $s = \scalingParameter{\ep}$, $d = \numberOfMessagesSmall{\ep}{n}$, and $p = \noiseProbabilitySmall{\ep}{n}$ (as well as the assumption $\ep > 1/{n}^{2/3}$), 
\begin{align*}
  & \frac{0.3\ep}{40 s^3} \cdot \Pr[i_0 \leq M \leq 2p(n-1)] - \Pr[M < i_0 ] \\
  & \geq \frac{0.3\ep}{160 s^3} - \exp(-1/(2\ep)) \\
  & \geq c\ep^4 - \exp(-1/(2\ep)),
\end{align*}
for some sufficiently small positive absolute constant $c$. The above quantity is positive as long as $\exp(1/(2\ep)) \geq \frac{1}{c\ep^4}$, i.e., as long as $\ep \leq c'$ for some absolute constant $c' > 0$ (which holds as long as we select $c_0 \leq c'$). From this and~\eqref{eq:sum-delta}, we can conclude that~\eqref{eq:target-ratio-first} holds.

\paragraph{Proof of (\ref{eq:target-ratio-second}).} Next, we move on to prove~\eqref{eq:target-ratio-second}. Similar to the previous case (specifically~\eqref{eq:final-borrow-right-first}), we may bound the right hand side of~\eqref{eq:expand-plus-one} further as
\begin{align}
&\sum_{i=0}^{n - 1} \binom{n - 1}{i} p^{i} (1 - p)^{n - 1 - i} \cdot P_{i, t - (n - i)\left(\frac{d - 1}{2}\right)} \\
&\geq \frac{1}{e^{\eps}} \cdot \sum_{i=0}^{n - 1} \binom{n - 1}{i} p^{i} (1 - p)^{n - 1 - i}  e^{\eps/2} \cdot  P_{i, t - (n - i) \left(\frac{d - 1}{2}\right)} \nonumber \\
&\qquad + \frac{1}{e^{\eps}} \cdot \sum_{i=0}^{n - 1} \binom{n - 1}{i} p^{i} (1 - p)^{n - 1 - i} \cdot \frac{p(n - 1 - i)}{i + 1} \cdot (e^{\eps/2} - 1) P_{i + 1, t - (n - i - 1)\left(\frac{d - 1}{2}\right)}.
\label{eq:final-borrow-right-second}
\end{align}
Thus, as in~\eqref{eq:sum-delta}, we may write the difference between the right hand side and the left hand side of~\eqref{eq:target-ratio-second} as
\begin{align}
&e^{\varepsilon} \cdot \Pr_{Z_1, \dots, Z_n \sim \cR_0}[Z_1 + \cdots + Z_n = t] - \Pr_{Z_1, \dots, Z_{n - 1} \sim \cR_0, Z_n \sim \cR_1}[Z_1 + \cdots + Z_n = t] \nonumber \\
&\geq (1 - p) \sum_{i=0}^{n - 1} \binom{n - 1}{i} p^{i} (1 - p)^{n - 1 - i} \tDelta_{i}, \label{eq:sum-delta-second}
\end{align}
where 
\begin{align*}
\tDelta_i &:= p(1 - e^{-\eps/2}) \cdot \left(P_{i + 1, t - (n - 1 - i)\left(\frac{d - 1}{2}\right)} + \frac{n - 1 - i}{i + 1} \cdot P_{i + 1, t - (n - i - 1)\left(\frac{d - 1}{2}\right)} \right) \\
& \qquad + e^{\eps/2} \cdot P_{i, t - (n - i) \left(\frac{d - 1}{2}\right)} - P_{i, t - (n - i)\left(\frac{d - 1}{2}\right)-1} .
\end{align*}
To see that the expression (\ref{eq:sum-delta-second}) is non-negative, observe first that, due to symmetry, we have $P_{i', j'} = P_{i', di' - j'}$ for all $i' \in \Z_{\geq 0}$ and $j \in \Z$. In particular,
\begin{align*}
  P_{i, t-(n-i) \left( \frac{d-1}{2} \right)} &= P_{i, di - \left( t-(n-i) \left( \frac{d-1}{2} \right) \right)} \\
  P_{i+1, t-(n-1-i) \left( \frac{d-1}{2} \right)} &= P_{i+1, di - \left( t - (n+1-i) \left( \frac{d-1}{2} \right)\right) + 1}.
\end{align*}
Using this observation together with Lemma~\ref{lem:tail-bound} where $m = i, k = di - \left(t - (n - i)\left(\frac{d - 1}{2}\right) - 1\right)$ and $\ell_1 = \ell_2 = \frac{d - 1}{2}$, we have that
\begin{equation}
  \label{eq:tdelta}
  \tDelta_i \geq (e^{0.3\ep} - 1) \cdot P_{i, di - \left(t - (n-i) \left( \frac{d-1}{2} \right)- 1\right)} = (e^{0.3\ep} - 1) \cdot P_{i,  t - (n-i) \left( \frac{d-1}{2} \right) -1 } \geq 0
\end{equation}for all $i_0 = \mLowerBound{\ep} \leq i \leq n-1$. Using Lemma \ref{lem:pia-lb} in a similar manner to the derivation of (\ref{eq:pimax-final}), we may conclude that for some $\tilde i_{\max} \in \{0, 1, \ldots, i_0 - 1 \}$, 
\begin{equation}
 (1 - p) \sum_{i=0}^{n - 1} \binom{n - 1}{i} p^{i} (1 - p)^{n - 1 - i} \tDelta_{i}  \geq P_{\tilde i_{\max}, t-(n-\tilde i_{\max}) \left( \frac{d-1}{2} \right) - 1} \cdot \left(\modified{\frac{0.3\ep}{40 s^3}} \cdot \Pr[i_0 \leq M \leq 2p(n-1)] - \Pr[M < i_0 ] \right)\nonumber.
\end{equation}
The same argument as in the proof of~\eqref{eq:target-ratio-first} establishes that as long as $c_0$ is chosen small enough, the above quantity is non-negative. It follows that~\eqref{eq:target-ratio-second} holds, and hence our proof is completed.
\end{proof}

\subsection{A Tale of Two Tails: Proof of Lemma~\ref{lem:tail-bound}}
\label{sec:tail-bound}

In this section we prove several inequalities relating the two tails $P_{m, \ast}$ and $P_{m + 1, \ast}$, and ultimately prove Lemma~\ref{lem:tail-bound}. Throughout this section, we will use the several additional notation:
\begin{itemize}
\item First, we will overload the notation and use $\nu(z)$ to denote the probability mass function of $\nu$ at $z$, i.e., $\nu(z) := \Pr_{Z \sim \nu}[Z = z]$.
\item We often represent a sequence of integers $a_1, \dots, a_m$ as a vector $\ba = (a_1, \dots, a_m)$; boldface is used to emphasized that the variable is a vector. For such a vector, we use $\nu(\ba)$ as a shorthand for the product $\nu(a_1) \cdots \nu(a_m)$.
\item We use $S_{m, k, d}$ to denote the set of all sequences of integers $a_1, \dots, a_m$ between $0$ and $d$ (inclusive) whose sum is $k$; more formally,
$$S_{m, k, d} = \{(a_1, \dots, a_m) \in (\Z \cap [0, d])^m \mid a_1 + \cdots + a_m = k\}
= \Delta_{m, k} \cap [0, d]^m.$$
Since $d$ will be fixed throughout, for simplicity of notation, we omit $d$ and simply use $S_{m, k}$. 
\item For a sequence $\ba = (a_1, \dots, a_m)$, we define $\zero(\ba)$ to be the number of zero coordinates, i.e., $\zero(\ba) = |\{i \in [m] \mid a_i = 0\}|$.
\item Next, for any $i \in \R$, we use $S_{m, k}^{\zero < i}$ (resp. $S_{m, k}^{\zero \geq i}$) to denote the sets of sequences in $S_{m, k}$ whose number of zero-coordinates is less than (resp., at least) $i$. More formally,
\begin{align*}
S_{m, k}^{\zero < i} = \{\ba \in S_{m, k} \mid \zero(\ba) < i\},
\end{align*}
and 
\begin{align*}
S_{m, k}^{\zero \geq i} = \{\ba \in S_{m, k} \mid \zero(\ba) \geq i\}.
\end{align*}
\end{itemize}

\paragraph{Proof Overview.} We now give a rough outline of our proof. First, let us observe that we may expand $P_{m, k}$ as
\begin{align*}
P_{m, k} = \sum_{\ba \in S_{m, k}} \nu(\ba) = \sum_{\ba \in S_{m, k}^{\zero < i}} \nu(\ba) + \sum_{\ba \in S_{m, k}^{\zero \geq i}} \nu(\ba),
\end{align*}
where $i$ will be chosen later in the proof.

We will bound the two terms on the right separately. More specifically, we will show that
\begin{align} \label{eq:overview-first-term-expansion}
\sum_{\ba \in S_{m, k}^{\zero < i}} \nu(\ba) \leq e^{0.5\eps} \cdot P_{m, k-1},
\end{align}
and that for $\ell_1, \ell_2 \in \{ \frac{d-1}{2}, \frac{d-3}{2} \}$, 
\begin{align} \label{eq:overview-second-term-expansion}
\sum_{\ba \in S_{m, k}^{\zero \geq i}} \nu(\ba) \leq p(1 - e^{-0.5\eps}) \cdot \left(P_{m+1, k+\ell_1}+ \frac{(n - m + 1)}{m + 1} \cdot P_{m + 1, k + \ell_2}\right). 
\end{align}
Once we have these two inequalities, Lemma~\ref{lem:tail-bound} immediately follows. The intuition behind the two inequalities is quite simple. For~\eqref{eq:overview-first-term-expansion}, since each sequence $\ba \in S_{m, k}^{\zero < i}$ contains few zeros, we should be able to pick a non-zero $a_i$ and decrease it by one and end up with a sequence in $S_{m, k - 1}$ instead; since the discrete Laplace distribution's mass (i.e., $\nu$) on $a_i$ and on $a_i - 1$ differs (multiplicatively) by a factor of at most $e^{1/s}$, the mass of the modified sequence also differs from the original sequence by a factor of $e^{1/s}$.

For~\eqref{eq:overview-second-term-expansion}, the intuition is pretty similar. We start with a sequence $\ba \in S_{m, k}^{\zero \geq i}$ and we will modify it to end up with a sequence in $S_{m + 1, k + \ell}$ where $\ell$ is either $\left(\frac{d - 1}{2}\right)$ or $\left(\frac{d - 3}{2}\right)$. The intuition here is that since $\ba$ contains many zero coordinates, there are many ways for us to divide $\ell$ among these zero coordinates and an additional coordinate, which would result naturally in a sequence in $P_{m + 1, k + \ell}$. 

To turn the intuition into a formal proof, we need to be careful about ``double counting'' a modified sequence. As an example, for~\eqref{eq:overview-first-term-expansion}, suppose we would like to modify a sequence in $S_{m, k}^{\zero < i}$ to one in $S_{m, k-1}$ by decreasing any non-zero coordinate. Then, it is possible that two sequences $(1, 0, a_3, \dots, a_n)$ and $(0, 1, a_3, \dots, a_n)$ results in the same sequence $(0, 0, a_3, \dots, a_n)$.

In order to avoid such ``double counting'', we divide our proofs into two parts. First, we show that we may replace $S_{m, k}^{\zero < i}$ (resp. $S_{m, k}^{\zero \geq i}$) with the set of sequences whose first coordinate is non-zero (resp., whose first few coordinates are zeros); this is done in Section~\ref{sec:permute}. Then, in Section~\ref{sec:prefix-mod}, we apply the modification step but only to the first (resp., first few) coordinates; this ensures that there is no ``double counting''. Finally, in Section~\ref{sec:final-tail-bound}, we put the two components together to deduce Lemma~\ref{lem:tail-bound}.

\subsubsection{Bounding Sums by (Non-)Zero Prefix Sums}
\label{sec:permute}

As stated earlier, we will show in this section that we may replace $S_{m, k}^{\zero < i}$ (resp., $S_{m, k}^{\zero \geq i}$) with the set of sequences whose first coordinate is non-zero (resp., whose first few coordinates are zeros). In both cases, the arguments are similar. Roughly speaking, we observe that permutations of coordinates of $\ba$ results in the same probability mass. Hence, by taking a random permutation of a sequence, there is a certain probability that we end up with a sequence with leading non-zero coordinate (resp., zero coordinates).

We can now formalize our bound, starting with that for $S_{m, k}^{\zero < i}$. Note that, for a permutation $\pi: [m] \to [m]$ and a sequence $\ba \in S_{m, k}$, we use $\pi \circ \ba$ to denote the sequence $(a_{\pi(1)}, \dots, a_{\pi(m)})$.

\begin{lemma} \label{lem:perm-bound-small}
For any $k \in \Z$ and $m, i \in \N$ such that $i \leq m$, we have
\begin{align*}
\sum_{\ba \in S_{m, k}^{\zero < i}} \nu(\ba) \leq \frac{m}{m - i + 1} \cdot \sum_{\ba' \in S_{m, k} \atop a'_1 \ne 0} \nu(\ba').
\end{align*}
\end{lemma}
\begin{proof}
We have
\begin{align*}
\sum_{\ba \in S_{m, k}^{\zero < i}} \nu(\ba) &\leq \sum_{\ba \in S_{m, k}^{\zero < i}} \left(\frac{1}{(m - i + 1) \cdot (m - 1)!} \cdot \sum_{\pi: [m] \to [m] \atop a_{\pi(1)} \ne 0} \nu(\ba)\right) \\
&= \frac{1}{(m - i + 1) \cdot (m - 1)!} \sum_{\ba \in S_{m, k}^{\zero < i}} \sum_{\pi: [m] \to [m] \atop a_{\pi(1)} \ne 0} \nu(\pi \circ \ba) \\
&= \frac{1}{(m - i + 1) \cdot (m - 1)!} \sum_{\ba' \in S_{m, k} \atop a'_1 = 0} \nu(\ba') \cdot \sum_{\ba \in S_{m,k}^{\zero<i}}|\{\pi \mid (\pi \circ \ba) = \ba'\}| \\
&\leq \frac{1}{(m - i + 1) \cdot (m - 1)!} \sum_{\ba' \in S_{m, k} \atop a'_1 \ne 0} \nu(\ba') \cdot m! \\
&\leq \frac{m}{m - i + 1} \sum_{\ba' \in S_{m, k} \atop a'_1 \ne 0} \nu(\ba'). 
\qedhere
\end{align*}
\end{proof}

We next prove our bound for $S_{m, k}^{\zero \geq i}$. In this case, we upper bound the sum $\sum_{\ba \in S_{m, k}^{\zero \geq i}} \nu(\ba)$ by the sum over sequences such that the first $t$ coordinates are zeros, where $t$ is a parameter that will be specified later.

\begin{lemma} \label{lem:perm-bound-large}
For any $k \in Z$ and $m, i, t \in \N$ such that $t \leq i \leq m$, we have
\begin{align*}
\sum_{\ba \in S_{m, k}^{\zero \geq i}} \nu(\ba) \leq \frac{m \cdot \cdots (m - t + 1)}{i \cdot \cdots (i - t + 1)} \cdot \nu(0)^t \cdot P_{m - t, k}.
\end{align*}
\end{lemma}

\begin{proof}
We have
\begin{align*}
\sum_{\ba \in S_{m, k}^{\zero \geq i}} \nu(\ba) &\leq \sum_{\ba \in S_{m, k}^{\zero \geq i}} \left(\frac{1}{i \cdots (i - t + 1) \cdot (m - t)!} \cdot \sum_{\pi: [m] \to [m] \atop a_{\pi(1)} = \cdots = a_{\pi(t)} = 0} \nu(\ba)\right) \\
&= \frac{1}{i \cdots (i - t + 1) \cdot (m - t)!} \sum_{\ba \in S_{m, k}^{\zero \geq i}} \sum_{\pi: [m] \to [m] \atop a_{\pi(1)} = \cdots = a_{\pi(t)} = 0} \nu(\pi \circ \ba) \\
&= \frac{1}{i \cdots (i - t + 1) \cdot (m - t)!} \sum_{\ba' \in S_{m, k} \atop a'_1 = \cdots = a'_t = 0} \nu(\ba') \cdot \sum_{\ba \in S_{m,k}^{\zero \geq i}}|\{\pi \mid (\pi \circ \ba) = \ba'\}| \\
&\leq \frac{1}{i \cdots (i - t + 1) \cdot (m - t)!} \sum_{\ba' \in S_{m, k} \atop a'_1 = \cdots = a'_t = 0} \nu(\ba') \cdot m! \\
&= \frac{m \cdots (m - t + 1)}{i \cdots (i - t + 1)} \sum_{\ba' \in S_{m, k} \atop a'_1 = \cdots = a'_t 0} \nu(\ba') \\
&= \frac{m \cdot \cdots (m - t + 1)}{i \cdot \cdots (i - t + 1)} \cdot \nu(0)^t \cdot P_{m - t, k}. \qedhere
\end{align*}
\end{proof}

\subsubsection{Bounding Sums by Prefix Modification}
\label{sec:prefix-mod}

We now move on to relate the sums derived in the previous sections to the terms that we actually care about (i.e., $P_{m, k - 1}, P_{m + 1, k + \left(\frac{d - 1}{2}\right)}, P_{m + 1, k + \left(\frac{d - 3}{2}\right)}$). As describe in the proof overview, this is done by modifying the first few coordinates of the sequences.

We start with the bound on the sum from Lemma~\ref{lem:perm-bound-small}. In this case, the modification is simple: just decrease the first coordinate by one. This is formalized below.

\begin{lemma} \label{lem:shift-single}
For any $m, k \in \N$, we have
\begin{align*}
\sum_{\ba' \in S_{m, k} \atop a'_1 \ne 0} \nu(\ba') \leq e^{1/s} \cdot P_{m, k - 1}.
\end{align*}
\end{lemma}
\begin{proof}
We can now further rewrite the right hand side as
\begin{align*}
\sum_{\ba' \in S_{m, k} \atop a'_1 \ne 0} \nu(\ba') = &\sum_{a'_1, \dots, a'_m \in \Z \cap [0, d] \atop a'_1 + \cdots + a'_m = k, a'_1 \geq 1} \nu(a'_1) \cdots \nu(a'_m) \\
&= \sum_{a''_1, a'_2, \dots, a'_{m - 1}, a'_m \in \Z \cap [0, d] \atop a''_1 + a'_2 + \cdots + a'_{m - 1} = k - 1, a''_1 \leq d - 1} \nu(a''_1 + 1) \nu(a'_2) \cdots \nu(a'_m) \\
&\leq e^{1/s} \cdot \left(\sum_{a''_1, a'_2, \dots, a'_{m - 1}, a'_m \in \Z \cap [0, d] \atop a''_1 + a'_2 + \cdots + a'_{m - 1} = k - 1, a''_1 \leq d - 1} \nu(a''_1) \nu(a'_2) \cdots \nu(a'_m)\right) \\
&\leq e^{1/s} \cdot \left(\sum_{a''_1, a'_2, \dots, a'_{m - 1}, a'_m \in \Z \cap [0, d] \atop a''_1 + a'_2 + \cdots + a'_{m - 1} = k - 1} \nu(a''_1) \nu(a'_2) \cdots \nu(a'_m)\right) \\
&= e^{1/s} \cdot P_{m, k - 1}. \qedhere
\end{align*}
\end{proof}

Next, for the right hand side term from Lemma~\ref{lem:perm-bound-large}, we will simply bound $\nu(0)^t$. In this case, the bound is shown by simply counting the number of possible ways of writing $\ell$ (which is either $\frac{d - 1}{2}$ or $\frac{d - 3}{2}$) as a sum of $t + 1$ non-negative integers, as stated more precisely below.

\begin{lemma} \label{lem:shift-mult}
Let $C = C_{d/2}(d/2, s)$. For any $t \in \N$ and any $\ell \in \Z \cap [0, d/2]$, we have
\begin{align*}
\frac{\binom{\ell + t}{t} \cdot e^{-(d/2 - \ell)/s}}{C} \cdot  \nu(0)^t = P_{t + 1, \ell}.
\end{align*}
\end{lemma}

\begin{proof}
For any $\ba \in S_{t + 1, \ell}$, we have
\begin{align}
\nu(\ba) &= \nu(a_1) \cdots \nu(a_{t + 1}) \nonumber \\
&= \left(\frac{1}{C} \cdot e^{-(d/2 - a_1)/s}\right) \cdots \left(\frac{1}{C} \cdot e^{-(d/2 - a_{t + 1})/s}\right) \nonumber \\
&= \left(\frac{1}{C} \cdot e^{-d/(2s)}\right)^t \cdot \left(\frac{1}{C} \cdot e^{-(d/2 - (a_1 + \cdots + a_{t + 1}))/s}\right) \nonumber  \\
&= \nu(0)^t \cdot \left(\frac{1}{C} \cdot e^{-(d/2 - \ell)/s}\right). \label{eq:zero-expand-single-term}
\end{align}

Now, observe that, from a standard star and bar argument, we have $|S_{t + 1, \ell}| = \binom{\ell + t}{t}$. As a result, we have  
\begin{align*}
P_{t + 1, \ell} &= \sum_{\ba \in S_{t + 1, \ell}} \nu(\ba) \\ 
&\overset{\eqref{eq:zero-expand-single-term}}{=} \sum_{\ba \in S_{t + 1, \ell}} \nu(0)^t \cdot \left(\frac{1}{C} \cdot e^{-(d/2 - \ell)/s}\right) \\ 
&= \frac{\binom{\ell + t}{t} \cdot e^{-(d/2 - \ell)/s}}{C} \cdot  \nu(0)^t,
\end{align*}
as desired.
\end{proof}

\subsubsection{Putting Things Together: Proof of Lemma~\ref{lem:tail-bound}}
\label{sec:final-tail-bound}

With the above four lemmas ready, we can now prove Lemma~\ref{lem:tail-bound} by picking appropriate values of $i, t$. To facilitate our proof, we will also employ the following lemma. 
\begin{lemma} \label{lem:sequence-separation}
For any $i, j, i', j' \in \N_0$, we have
\begin{align*}
P_{i, j} \cdot P_{i', j'} \leq P_{i + i, j + j'}.
\end{align*}
\end{lemma}
The proof of Lemma \ref{lem:sequence-separation} is deferred to Section \ref{sec:sequence-separation}.

\begin{proof}[Proof of Lemma~\ref{lem:tail-bound}]
Recall that we would like to show:
\begin{align} \label{eq:tail-bound-main}
e^{-\ep} p(1 - e^{-\eps/2}) \cdot \left(P_{m+1, k+\ell_1} + \frac{(n - m - 1)}{m + 1} \cdot P_{m + 1, k + \ell_2}\right) + e^{0.2\ep} \cdot P_{m, k - 1} \geq P_{m, k}
\end{align}
for all $\mLowerBound{\ep} \leq m \leq n-1$ when $p \geq \noiseProbabilityLB{\ep}{n}$.

Let $i = \lceil (1 - e^{0.1\eps}) \cdot m \rceil$.
We may write $P_{m, k}$ as
\begin{align} \label{eq:expand-by-num-zeros}
P_{m, k} = \sum_{\ba \in S_{m, k}} \nu(\ba)
= \sum_{\ba \in S_{m, k}^{\zero < i}} \nu(\ba) + \sum_{\ba \in S_{m, k}^{\zero \geq i}} \nu(\ba).
\end{align}
We will bound the two terms on the right hand side separately. 

\paragraph{First term of (\ref{eq:expand-by-num-zeros}).} For the first term (i.e., the sum over $\ba \in S_{m,k}^{\zero<i}$), by applying Lemmas~\ref{lem:perm-bound-small} and~\ref{lem:shift-single}, we have
\begin{align} \label{eq:first-zero-first-ineq}
\sum_{\ba \in S_{m, k}^{\zero < i}} \nu(\ba) &\leq e^{1/s} \cdot \frac{m}{m - i + 1} \cdot P_{m, k - 1}.
\end{align}
Recall that we pick $s$ so that $1/s \leq 0.1\varepsilon$ and $i$ so that $\frac{m}{m - i + 1} \leq \frac{m}{m - (1 - e^{-0.1\eps})m} = e^{0.1\eps}$. Combining these two inequalities with~\eqref{eq:first-zero-first-ineq}, we have
\begin{align} \label{eq:first-zero}
\sum_{\ba \in S_{m, k}^{\zero < i}} \nu(\ba) &\leq e^{0.2\varepsilon} \cdot P_{m, k-1}. 
\end{align}

\newcommand{\tlogn}{\left\lceil 100\log\left( \frac{n}{1-e^{-0.1\ep}}\right) \right\rceil}
\newcommand{\tzeros}{\min\left\{\lceil i/2 \rceil, \tlogn\right\}}

\paragraph{Second term of (\ref{eq:expand-by-num-zeros}).} We now move on to bound the second term on the right hand side of~\eqref{eq:expand-by-num-zeros}. For this term, we apply Lemma~\ref{lem:perm-bound-large} and Lemma~\ref{lem:shift-mult} with $t =\tzeros$. This gives the following for any $\ell \in \left\{\frac{d - 1}{2}, \frac{d - 3}{2}\right\}$ and $0 \leq m \leq n-1$:
\begin{align}
\sum_{\ba \in S_{m, k}^{\zero \geq i}} \nu(\ba) 
&\leq \frac{m \cdot \cdots (m - t + 1)}{i \cdot \cdots (i - t + 1)} \cdot \nu(0)^t \cdot P_{m - t, k} \nonumber \\
&\leq \frac{m \cdot \cdots (m - t + 1)}{i \cdot \cdots (i - t + 1)} \cdot \frac{C}{\binom{\ell + t}{t} \cdot e^{-(d/2 - \ell)/s}} \cdot P_{t + 1, \ell} \cdot P_{m - t, k} \nonumber \\
&\leq \left(\frac{m}{i - t + 1}\right)^t \cdot \left(\frac{t}{\ell}\right)^t \cdot (C \cdot e^{(d/2 - \ell)/s}) \cdot P_{t + 1, \ell} \cdot P_{m - t, k} \nonumber \\
&\leq \left(\frac{m}{i - t + 1}\right)^t \cdot \left(\frac{t}{\ell}\right)^t \cdot (C \cdot e^{(d/2 - \ell)/s}) \cdot P_{m + 1, k + \ell} \label{eq:almost-final-many-zeros},
\end{align}
where the last inequality follows from Lemma~\ref{lem:sequence-separation}. 

Now, from $\ell \in \left\{\frac{d - 1}{2}, \frac{d - 3}{2}\right\}$ and from Lemma~\ref{lem:normalization-val}, we have
\begin{align} \label{eq:constant-term-bound}
(C \cdot e^{(d/2 - \ell)/s}) \leq \frac{2}{1 - e^{-0.1\eps}} \cdot e^{3/(2s)} = \frac{2e^{0.15\eps}}{1 - e^{-0.1\eps}}.
\end{align}

Next, from our choice of $t = \tzeros$ and since $d \geq \numberOfMessagesSmall{\ep}{n}$ holds for all $0 \leq m \leq n-1$, we have
\begin{align}
\left(\frac{m}{i - t + 1}\right)^t \cdot \left(\frac{t}{\ell}\right)^t &\leq \left(\frac{m}{i/2}\right)^t \cdot \left(\frac{t}{\ell}\right)^t \nonumber\\
(\text{from } i \geq (1 - e^{-0.1\eps}) m) &\leq \left(\frac{2t}{(1 - e^{-0.1\eps}) \ell}\right)^t \nonumber\\
\label{eq:pre-mid-bound}
  (\text{from } \ell \in \left\{(d - 1)/2, (d - 3)/2\right\}) &\leq \left(\frac{2t}{(1 - e^{-0.1\eps}) \left(\frac{d - 3}{2}\right)}\right)^t \\
  \label{eq:mid-bound}
(\text{from  our choice of } d) &\leq \left(\frac{t e^{-100\eps}}{1000\log(n/(1-e^{-0.1\ep}))}\right)^t.
\end{align}

Let us now consider two cases, based on whether $t = \tlogn$. 

\paragraph{Case 1: $t = \tlogn$.} In this case, we have from (\ref{eq:mid-bound}), 
\begin{align*}
  \left(\frac{m}{i - t + 1}\right)^t \cdot \left(\frac{t}{\ell}\right)^t &\leq \left(\frac{\left(100 \log\left(\frac{n}{1-e^{-0.1\ep}}\right) + 1\right)e^{-100\eps}}{1000\log(n/(1-e^{-0.1\ep}))}\right)^{t} \\
                                                                         & \leq e^{-100\ep} \cdot 2^{-t-1} \\
  & \leq e^{-100\ep} \left( \frac{1-e^{-0.1\ep}}{n} \right)^{100} \cdot \frac 12\\
&\leq e^{-\ep} (1 - e^{-0.1\eps})^2 e^{-100\eps} \cdot \frac p2,
\end{align*}
where the final inequality follows since $p \geq \noiseProbabilityLB{\ep}{n}$. 

Combining the above inequality with~\eqref{eq:constant-term-bound}, \eqref{eq:almost-final-many-zeros}, we have that for all $0 \leq m \leq n-1$,
\begin{align*}
\sum_{\ba \in S_{m, k}^{\zero \geq i}} \nu(\ba) \leq (1 - e^{-0.1\eps}) \cdot e^{-\ep} p \cdot P_{m+1, k+\left(\frac{d - 1}{2}\right)}.
\end{align*}

\paragraph{Case 2: $t \ne \tlogn$.} From our choice of $t$, we must have $t = \lceil i/2 \rceil$ and $t \leq \tlogn$. From our assumption on $m$, it follows that
$$
t \geq i/2 \geq \frac{m(1-e^{-0.1\ep})}{2} \geq 5 \log \left( \frac{1}{1-e^{-0.1\ep}}\right).
$$

Then by (\ref{eq:mid-bound})
\begin{align}
\left(\frac{m}{i - t + 1}\right)^t \cdot \left(\frac{t}{\ell}\right)^t 
&\leq  \left(\frac{t e^{-100\eps}}{1000\log(n/(1-e^{-0.1\ep}))}\right)^t \nonumber\\ 
&\leq e^{-100\eps} \cdot 4^{-t-1} \nonumber \\
  &\leq (1 - e^{-0.1\eps})^{5} e^{-100\eps} / (4t) \nonumber\\
  &\leq (1-e^{-0.1\ep})^2 e^{-100\ep}/(4t). \label{eq:tmp-exp}
\end{align}
As $t = \lceil i/2 \rceil$, we have that $i \leq 2t$ and so $m \leq \frac{i}{1 - e^{-0.1\eps}} \leq \frac{2t}{1 - e^{-0.1\eps}}$.
Now, recall our assumption that $\eps > \frac{1}{{n}^{2/3}}$ (which holds for all $0 \leq m \leq n-1$). This means that $m \leq O(t{n}^{2/3}) \leq O({n}^{2/3} \log n)$.  Hence, for any sufficiently large $n$, we must have $m \leq n/2 - 1$. Thus, we have
\begin{align} \label{eq:term-from-next-binomial}
\frac{p(n - 1 - m)}{m + 1} \geq \frac{pn}{4m} \geq \frac{10}{(1 - e^{-0.1\eps}) m} \geq \frac{1}{t},
\end{align}
where the second-to-last inequality comes from $p \geq \noiseProbabilityLB{\ep}{n}$ and the last inequality comes from $m \leq \frac{2t}{1 - e^{-0.1\eps}}$. 
As a result, by combining~\eqref{eq:tmp-exp} and~\eqref{eq:term-from-next-binomial}, we obtain
\begin{align}
  \label{eq:pcasefinal}
\left(\frac{m}{i - t + 1}\right)^t \cdot \left(\frac{t}{\ell}\right)^t \leq (1 - e^{-0.1\eps})^2 e^{-100\eps} \cdot \frac{p(n - 1 - m)}{m + 1}.
\end{align}

By \eqref{eq:pcasefinal}, together with~\eqref{eq:constant-term-bound} and~\eqref{eq:almost-final-many-zeros}, we have that for $\ell \in \left\{ \frac{d-1}{2}, \frac{d-3}{2} \right\}$,
\begin{align*}
\sum_{\ba \in S_{m, k}^{\zero \geq i}} \nu(\ba) \leq e^{-\ep} (1 - e^{-0.1\eps}) \cdot \frac{p(n - 1 - m)}{m + 1} \cdot P_{m+1, k+\ell}.
\end{align*}

Thus, in both cases 1 and 2 we consider, we have, for $\ell_1, \ell_2 \in \left\{ \frac{d-1}{2}, \frac{d-3}{2} \right\}$, and the claimed values of $m,p$,
\begin{align*}
\sum_{\ba \in S_{m, k}^{\zero \geq i}} \nu(\ba) &\leq e^{-\ep} p(1 - e^{-0.1\eps}) \cdot \left(P_{m+1, k+\ell_1}+ \frac{p(n - m + 1)}{m + 1} \cdot P_{m + 1, k + \ell_2}\right) \\
&\leq e^{-\ep} p(1 - e^{-0.5\eps}) \cdot \left(P_{m+1, k+\ell_1}+ \frac{p(n - m + 1)}{m + 1} \cdot P_{m + 1, k + \ell_2}\right).
\end{align*}

Combining this with~\eqref{eq:expand-by-num-zeros} and~\eqref{eq:first-zero} yields the claimed bound.
\end{proof}

\section{Lower Bound for Binary Summation}\label{sec:lb_bit_agg}

In this section we prove our lower bound on the communication complexity of any non-interactive pure-$\shuffledDP$ protocol that can perform bit addition with small error. Specifically, we show that any $O(1)$-$\shuffledDP$ protocol must have communication complexity at least $\Omega(\sqrt{\log n})$. In fact, as formalized below, our lower bound holds even against any protocol that has an expected error of $O(n^{0.5 - \Omega(1)})$. Recall that the standard randomized response, which is an $e^{\eps}$-$\localDP$ protocol, incurs an error of $O_{\eps}(n^{0.5})$ and has communication complexity of only one bit. Thus, our lower bound states that, even to slightly improve upon this simple pure-DP protocol in terms of error, the communication complexity must blow up to $\Omega(\sqrt{\log n})$.

\begin{theorem} \label{thm:lb-main}
For any constants $\eps > 0$ and $\chi > 0$, there is no $\eps$-$\shuffledDP$ non-interactive protocol with communication complexity $o(\sqrt{\log n})$ that incurs $O\left(n^{0.5 - \chi}\right)$ error.
\end{theorem}

We remark that Cheu et al.~\cite{CheuSUZZ19} proved that, with appropriate setting of parameters, the simple randomized response is an $(\eps, \delta)$-$\shuffledDP$ protocol and incurs an expected error of at most $O\left(\frac{\eps^2}{\log(1/\delta)}\right)$. Since the user's communication in their protocol is just a bit, our result also gives a communication complexity separation between pure-$\shuffledDP$ and approximate-$\shuffledDP$.

Another remark here is that our lower bound in Theorem~\ref{thm:lb-main} is roughly a square of the upper bound $O(\log n)$ obtained in our protocol for the previous section (for constant values of $\epsilon$). It remains an interesting open question to close this $O(\sqrt{\log n})$ gap. On this front, we will show in Section~\ref{sec:MGF_limit} that, for our specific approach,  $O(\sqrt{\log n})$ lower bound is the best one could hope for, which means that our lower bound in Theorem~\ref{thm:lb-main} is tight for the current approach.

We first recall the following standard notion from probability theory.
\begin{definition} [Moment Generating Function]
\label{def:mgf}
Let $\bY$ be a random variable supported on (a subset of) $\R^k$ for some $k \in \N$. Its \emph{moment generating function (MGF)} is defined as $\bM_{\bY}(\bt) = \E[e^{\left<\bt, \bY\right>}]$.
\end{definition}
Throughout this section, we will be dealing with pairs of random variables whose MGFs are within a certain factor of each other. The following definition will be particularly handy.
\begin{definition} [Bounded MGF ratio]
\label{def:mgf-bounded-ratio}
We say that two random variables $\bY, \bY'$ supported on (a subset of) $\R^k$ have \emph{$e^{\eps}$-bounded MGF ratio} if and only if, for all $t \in \mathbb{R}^k$ we have that $\frac{\bM_{\bY}(\bt)}{\bM_{\bY'}(\bt)} \in [e^{-\eps}, e^\eps]$. \end{definition}
For two random variables $\bY, \bY'$, let $\SD(\bY, \bY')$ denote the \emph{total variation distance} between them.

Our proofs follow exactly the same outline as in Section~\ref{sec:technique-overview}. Specifically, the remainder of this section is organized as follows. In Section~\ref{sec:mgf}, we prove that a pure-$\shuffledDP$ protocol implies bounded MGF ratio condition. Then, in Section~\ref{sec:mgf-to-comm-bound}, we give a lower bound on $C_{\eps, \gamma}$ from Definition~\ref{def:lower-bound-question} and use it to prove our main theorem of this section (Theorem~\ref{thm:lb-main}). Finally, in Section~\ref{sec:MGF_limit}, we provide an example which shows that our lower bound for the question is tight.

\begin{remark}
 The lower bound of Theorem~\ref{thm:lb-main} has been stated for non-interactive protocols in the shuffled model that are \emph{symmetric}, i.e., protocols for which the local randomizer (given by $\bX^0$ and $\bX^1$ from Definition~\ref{def:non-interactive-protocols}) is identical for each user. However, the lower bound actually generalizes to protocols that are not necessarily symmetric (and in which the number of messages can vary from user to user). Indeed, one can show that it is not possible to obtain error $O(n^{0.5-\chi})$ unless, for at least $1 - o(1)$ fraction of the users, the communication complexity is $\Omega(\sqrt{\log n})$. We have omitted the formal statement for the sake of clarity of exposition, but the proof is almost identical, as the $e^{\epsilon}$-bounded MGF property (given by Lemma~\ref{lem:dp-mgf-ratio}) holds for \emph{any} user's $\bX^0$, $\bX^1$ (this can be seen by comparing two sequences that differ in the given user's input), and Theorem~\ref{thm:chan-local} also applies to the asymmetric case (with the guarantee that $1 - o(1)$ fraction of the users must have $\SD(\bX^0, \bX^1) \geq 1 - n^{-\Omega(1)}$).
\end{remark}

\subsection{Pure-DP Implies MGF Bounded Ratio}
\label{sec:mgf}

In this subsection, we will prove a general necessary (but not sufficient) condition on $\varepsilon$-DP protocols in terms of the MGFs of $\bX^0, \bX^1$. 
%
A straightforward observation we will use is the following:
\begin{observation} \label{obs:mgf-ratio}
Let $\bY, \bY'$ be two random variables with the same support $\supp(\bY) = \supp(\bY') \subseteq \mathbb{R}^k$ such that $\frac{\Pr[\bY = \bv]}{\Pr[\bY' = \bv]} \in [e^{-\eps}, e^{\eps}]$. Then, $\bY, \bY'$ satisfies $e^{\eps}$-bounded MGF ratio.
\end{observation}

\begin{proof}
Consider any $\bt \in \R^k$. We have
$$\frac{\bM_{\bY_1}(\bt)}{\bM_{\bY_2}(\bt)} = \frac{\sum_{\by \in \R^m} \Pr[\bY_1 = \by] \cdot e^{\left<\bt, \by\right>}}{\sum_{\by \in \R^m} \Pr[\bY_2 = \by] \cdot e^{\left<\bt, \by\right>}}.$$
From our assumption, each ratio of the corresponding terms on the RHS lies in $[e^{-\varepsilon}, e^{\varepsilon}]$. Hence, we can conclude that $\bM_{\bY_1}(\bt)/\bM_{\bY_2}(\bt) \in [e^{-\varepsilon}, e^{\varepsilon}]$ as desired.
\end{proof}

In general, the converse of the above is not true, i.e., there are pairs of distributions whose probability ratios are not within the desired range but the MGF ratios are within the range (e.g., the distributions from our randomizer in the previous section). Nonetheless, we can show that, for any $\varepsilon$-DP protocol, $\bX^0, \bX^1$ must satisfy the weaker condition of $e^{\eps}$-bounded MGF ratio, as stated below. This is our main observation.

\begin{lemma} \label{lem:dp-mgf-ratio}
For any $\varepsilon$-DP protocol, $\bX^0, \bX^1$ must satisfy $e^{\eps}$-bounded MGF ratio.
\end{lemma}

To prove Lemma~\ref{lem:dp-mgf-ratio}, a key (well-known) multiplicative property of MGF that we need is that, if $\bY, \bY' \in \mathbb{R}^k$ are two independent random variables, then $\bM_{\bY + \bY'}(\bt) = \bM_{\bY}(\bt) \cdot \bM_{\bY'}(\bt)$ for all $\bt \in \mathbb{R}^k$. We this in mind, we can prove Lemma~\ref{lem:dp-mgf-ratio} as follows.

\begin{proof}[Proof of Lemma~\ref{lem:dp-mgf-ratio}]
Consider two sequences $0 \dots 00$ and $0 \dots 0 1$, each of length $n$. Let $\bY^0, \bY^1 \in \R^{k}$ denote the views of the shuffled output on the corresponding input vectors, where $\bY^0_j$ denote the number of $j$'s received by the analyzer for the input vector $0 \ldots 00$ and $\bY^1_j$ denote the number of $j$'s received by the analyzer for the input vector $0 \ldots 01$. Notice that $\bY^0$ is simply a sum of $n$ i.i.d. copies of $\bX^0$ and $\bY^1$ is a sum of $(n - 1)$ i.i.d. copies of $\bX^0$ and a copy of $\bX^1$. Observe also that $\varepsilon$-DP implies that $\bY^0, \bY^1$ satisfy the condition in Observation~\ref{obs:mgf-ratio}. From this, we have
\begin{align} \label{eq:div-mgf}
[e^{-\varepsilon}, e^{\varepsilon}] \ni \frac{\bM_{\bY^0}(\bt)}{\bM_{\bY^1}(\bt)} = \frac{(\bM_{\bX^0}(\bt))^n}{(\bM_{\bX^0}(\bt))^{n - 1} \cdot \bM_{\bX^1}(\bt)} = \frac{\bM_{\bX^0}(\bt)}{\bM_{\bX^1}(\bt)},
\end{align}
for all $\bt \in \mathbb{R}^k$. This completes our proof.
\end{proof}

\subsection{From MGF Bounded Ratio to Communication Lower Bound}
\label{sec:mgf-to-comm-bound}

We will now use the MGF bounded ratio property from Lemma~\ref{lem:dp-mgf-ratio} to show the communication complexity of any non-interactive protocol for summation that incurs small error. To do so, let us recall below a known result that any protocol that can perform binary summation to within a small error must have large statistical distance between $\bX^0$ and $\bX^1$. (In fact, the bound below holds even for $\localDP$ protocols.)

\begin{theorem}[\cite{ChanSS12}] \label{thm:chan-local}
Any non-interactive protocol that can perform binary summation to within an expected absolute error of $\alpha$ (even in the local model) must satisfy $\SD(\bX^0, \bX^1) \geq 1 - O\left(\frac{\alpha}{\sqrt{n}}\right)$.
\end{theorem}

Note that Theorem~\ref{thm:chan-local} is not inherently about privacy, but rather about the utility and the output distributions. We remark that the above fact was implicitly first shown in~\cite{ChanSS12} under a slightly different terminology. For completeness, we provide a full proof of Theorem~\ref{thm:chan-local} in Appendix~\ref{app:tv-vs-utility-proof}


Thanks to Lemma~\ref{lem:dp-mgf-ratio} and Theorem~\ref{thm:chan-local}, to prove our lower bound (Theorem~\ref{thm:lb-main}), it now suffices to show that, for any pair of random variables $\bY, \bY'$ whose supports lie in $\Delta_{k, m}$ that satisfies both $e^{\eps}$-bounded MGF ratio and if $\SD(\bY, \bY')$ is large, then $m \log k$ must be large. The main lemma of this subsection, which encapsulates a quantitative version of the aforementioned statement, is stated formally below.

\begin{lemma}\label{lem:bounded-moment-ratio-to-sd}
Let $\bY, \bY'$ be two random variables supported on $\Delta_{k, m}$ with $e^{\eps}$-bounded MGF ratio. Then, 
\begin{align*}
\SD(\bY, \bY') \leq 1 - 2^{-O_{\eps}(m^2 \log k)}.
\end{align*}
\end{lemma}

Before we prove Lemma~\ref{lem:bounded-moment-ratio-to-sd}, we note that plugging together Lemma~\ref{lem:bounded-moment-ratio-to-sd}, Lemma~\ref{lem:dp-mgf-ratio}, and Theorem~\ref{thm:chan-local} immediately gives Theorem~\ref{thm:lb-main}, as follows.

\begin{proof}[Proof of Theorem~\ref{thm:lb-main}]
Consider any $\eps$-$\shuffledDP^1$ protocol that performs binary summation to within an expected absolute error of $\alpha := O(n^{0.5 - \gamma})$. From Observation~\ref{obs:mgf-ratio}, $\bX^0, \bX^1$ must satisfy $e^{\eps}$-MGF bounded ratio. Applying Lemma~\ref{lem:bounded-moment-ratio-to-sd} implies that
\begin{align*}
\SD(\bX^0, \bX^1) \leq 1 - 2^{-O_{\eps}(m^2 \log k)}.
\end{align*}
Furthermore, since the expected error of the protocol is at most $\alpha = O(n^{0.5 - \chi})$, Theorem~\ref{thm:chan-local} implies that
\begin{align*}
\SD(\bX^0, \bX^1) \geq 1 - O\left(\frac{\alpha}{\sqrt{n}}\right) = 1 - O\left(\frac{1}{n^{\chi}}\right).
\end{align*}
Combining the above two inequalities, we must have $m^2 \log k \geq \Omega_{\eps, \chi}(\log n)$, which implies that the communication complexity $m \log k$ must be at least $\Omega_{\eps, \chi}(\sqrt{\log n})$ as desired.
\end{proof}

\paragraph{Dual Approach and Proof of Lemma~\ref{lem:bounded-moment-ratio-to-sd}.} We devote the rest of this subsection to the proof of Lemma~\ref{lem:bounded-moment-ratio-to-sd}.
For notational convenience, we use $p_{\by}$ and $p'_{\by}$ to denote $\Pr[\bY = \by]$ and $\Pr[\bY' = \by]$ respectively. 

Before we formalize the proof below, let us first present an informal overview of the proof. 
Recall that $1 - \SD(\bY, \bY')$ is equal to $\sum_{\by \in \Delta_{k, m}} \min\{p_{\by}, p'_{\by}\} = $ $\min_{S \subseteq \Delta_{k, m}} \left\{\sum_{\by \in S} p_{\by} + \sum_{\by \in \Delta_{k, m} \setminus S} p'_{\by} \right\}$. Hence, it suffices for us to show that, for every $S \subseteq \Delta_{k, m}$, we have 
\begin{align} \label{eq:obj-small}
\sum_{\by \in S} p_{\by} + \sum_{\by \in \Delta_{k, m} \setminus S} p'_{\by} \geq 2^{-O_{\eps}(m^2 \log k)}.
\end{align}
We will give a ``dual certificate'' for this statement. Notice that since the total probability of each of $\bY, \bY'$ must be one, we have $\sum_{\by \in \Delta_{k, m}} p_{\by} = 1$ and $\sum_{\by \in \Delta_{k, m}} p'_{\by} = 1$.
Of course, we also have the non-negativity constraints that $p_{\by}, p'_{\by} \geq 0$ for all $\by \in \Delta_{k, m}$.

Furthermore, the $e^{\eps}$-bounded MGF ratio property between $\bY$ and $\bY'$ simply translates to the following linear inequalities for all $\bt \in \R^k$:
\begin{align} \label{eq:moment-raio-ineq-tmp1}
\sum_{\by \in \Delta_{k, m}} e^{\left<\bt, \by\right>} \cdot p'_{\by} - \sum_{\by \in \Delta_{k, m}} e^{\left<\bt, \by\right> - \eps} \cdot p_{\by} \geq 0,
\end{align}
and 
\begin{align} \label{eq:moment-raio-ineq-tmp2}
\sum_{\by \in \Delta_{k, m}} e^{\left<\bt, \by\right>} \cdot p_{\by} - \sum_{\by \in \Delta_{k, m}} e^{\left<\bt, \by\right> - \eps} \cdot p'_{\by} \geq 0.
\end{align}

Hence, we simply have a system of infinite) linear inequalities and we would like to certify a particular linear inequality~\eqref{eq:obj-small}. We may do this by writing~\eqref{eq:obj-small} as a linear combination of the constraints. 

As a wishful thinking, if we could somehow ``extract'' only the $p_{\by}$ and $p'_{\by}$ terms from~\eqref{eq:moment-raio-ineq-tmp1} and~\eqref{eq:moment-raio-ineq-tmp2}, then we would be done because we would simply have $e^{\eps} \cdot p_{\by} \geq p'_{\by} \geq e^{-\eps} \cdot p_{\by}$ which can easily be combined with the total probability and non-negativity constraints to get a good bound on $\sum_{\by \in S} p_{\by} + \sum_{\by \in \Delta_{k, m} \setminus S} p'_{\by}$. Of course, such extraction is not possible since, for any value $\bt$ we plug into~\eqref{eq:moment-raio-ineq-tmp1} and~\eqref{eq:moment-raio-ineq-tmp2}, we always get non-zero coefficients for all vectors in $\Delta_{k, m}$, not just $\by$.

With the above in mind, our goal is now to select one $\bt = \tau(\by)$ for each $\by$ in such a way that the coefficient of $\by$ from its own inequality (i.e., $\bt = \tau(\by)$) ``dominates'' the coefficients of $\by$ from other inequalities (i.e., $\bt = \tau(\by')$ for any $\by' \ne \by$). A more precise version of the statement is proved below. Note here that $e^{\beta(\by)}$ here should be thought of as the ``scaling factor'' for the inequality for $\by$.

\begin{lemma} \label{lem:coeff-selection}
For any $\eps > 0$, there exists a mapping $\tau: \Delta_{k, m} \to \mathbb{R}^k$ and $\beta: \Delta_{k, m} \to \mathbb{R}$ such that the following hold for all $\by \in \Delta_{k, m}$:
\begin{align} \label{eq:range-coeff-exponent}
0 \geq \left<\tau(\by), \by\right> + \beta(\by) \geq \zeta := -O_{\eps}(m^2 \log k),
\end{align}
and
\begin{align} \label{eq:coeff-dominate}
e^{\left<\tau(\by), \by\right>  + \beta(\by)} \geq 2e^{\eps} \cdot \sum_{\by' \in \Delta_{k, m} \setminus \{\by\}} e^{\left<\tau(\by'), \by\right>  + \beta(\by')}.
\end{align}
\end{lemma}

\begin{proof}
Let $\rho = \eps + 10 \ln(k + 1) + 10$.
We pick $\tau(\by) = \rho \cdot 2\by$ and $\beta(\by) = \rho \cdot \left(-\|\by\|_2^2 - m^2\right)$. It is obvious to see that~\eqref{eq:range-coeff-exponent} holds.
Next, to prove~\eqref{eq:coeff-dominate}, let us first observe the following identity:
\begin{align}
\left<\tau(\by'), \by\right> + \beta(\by') &= \rho\left(2\left<\by', \by\right> - \|\by'\|^2 - m^2\right) \nonumber \\
&= \rho\left(\|\by\|_2^2 - \|\by - \by'\|_2^2 - m^2\right) \nonumber \\
&= \left<\tau(\by), \by\right> + \beta(\by) - \rho \cdot \|\by - \by'\|_2^2. \label{eq:quad-form}
\end{align}
We may bound the right hand side of~\eqref{eq:coeff-dominate} as
\begin{align}
\sum_{\by' \in \Delta_{k, m} \setminus \{\by\}} e^{\left<\tau(\by'), \by\right>  + \beta(\by')} &\overset{\eqref{eq:quad-form}}{=} e^{\left<\tau(\by), \by\right>  + \beta(\by)} \cdot \sum_{\by' \in \Delta_{k, m} \setminus \{\by\}} e^{-\rho \|\by - \by'\|_2^2} \nonumber \\
&= e^{\left<\tau(\by), \by\right>  + \beta(\by)} \cdot \left( \sum_{i=1}^{2m^2} e^{-\rho i} \cdot |\{\by' \in \Delta_{k, m} \mid \|\by - \by'\|_2^2 = i\}|\right). \label{eq:sum-quad-expand}
\end{align}
We can bound $|\{\by' \in \Delta_{k, m} \mid \|\by - \by'\|_2^2 = i\}|$ as follows.
\begin{align}
|\{\by' \in \Delta_{k, m} \mid \|\by - \by'\|_2^2 = i\}|
&\leq |\{\bz \in \Z^k \mid \|\bz\|_2^2 = i\}| \nonumber \\
&\leq 2^i \cdot |\{\bz \in \Z^k \mid \|\bz\|_2^2 = i, z_1, \dots, z_k \geq 0\}| \nonumber \\
&\leq 2^i \cdot |\{(x_1, \dots, x_k) \in \Z^k_{\geq 0} \mid x_1 + \cdots + x_k = i\}| \nonumber \\
&= 2^i \cdot \binom{k + i - 1}{i} \nonumber \\
&\leq 2^i \left(\frac{e(k + i - 1)}{i}\right)^i \nonumber \\
&\leq (2e(k + 1))^i, \label{eq:int-vector-length-bound}
\end{align}
where the second inequality comes from the fact that there are at most $i$ non-zero coordinates of $\bz$ and there are two choices of sign for those coordinates.

Plugging~\eqref{eq:int-vector-length-bound} into~\eqref{eq:sum-quad-expand}, we have
\begin{align*}
\sum_{\by' \in \Delta_{k, m} \setminus \{\by\}} e^{\left<\tau(\by'), \by\right>  + \beta(\by')} &\leq e^{\left<\tau(\by), \by\right>  + \beta(\by)} \cdot \left( \sum_{i=1}^{2m^2} \left(e^{-\rho} \cdot 2e(k + 1)\right)^i \right) \\
(\text{From our choice of } \rho) &\leq e^{\left<\tau(\by), \by\right>  + \beta(\by)} \cdot \left( \sum_{i=1}^{2m^2} \left(\frac{1}{10e^{\eps}}\right)^i \right) \\
&\leq e^{\left<\tau(\by), \by\right> + \beta(\by)} \cdot \frac{1}{2 e^{\eps}},
\end{align*}
as desired.
\end{proof}

With Lemma~\ref{lem:coeff-selection} ready, we can now prove Lemma~\ref{lem:bounded-moment-ratio-to-sd}.

\begin{proof}[Proof of Lemma~\ref{lem:bounded-moment-ratio-to-sd}]
Let $\bY, \bY'$ be two random variables supported on (subsets of ) $\Delta_{k, m}$. Suppose that $\bY, \bY'$ satisfy $e^{\eps}$-bounded MGF ratio. Let $\tau, \beta$ be as in Lemma~\ref{lem:coeff-selection}.

Consider any set $S \subseteq \Delta_{k, m}$. For every $\by' \in S$, $\bM_{\bY}(\tau(\by')) \geq e^{-\eps} \cdot \bM_{\bY'}(\tau(\by'))$ is equivalent to
\begin{align} \label{eq:ineq-in-S}
\sum_{\by \in \Delta_{k, m}} e^{\left<\tau(\by'), \by\right> -  \beta(\by')} \cdot p_{\by} - \sum_{\by \in \Delta_{k, m}} e^{\left<\tau(\by'), \by\right> -  \beta(\by') - \eps} \cdot p'_{\by} \geq 0.
\end{align}
Similarly, for every $\by' \in \Delta_{k, m} \setminus S$, $\bM_{\bY'}(\tau(\by')) \geq e^{-\eps} \cdot \bM_{\bY}(\tau(\by'))$ can be rearranged as
\begin{align} \label{eq:ineq-outside-S}
\sum_{\by \in \Delta_{k, m}} e^{\left<\tau(\by'), \by\right> -  \beta(\by')} \cdot p'_{\by} - \sum_{\by \in \Delta_{k, m}} e^{\left<\tau(\by'), \by\right> -  \beta(\by') - \eps} \cdot p_{\by} \geq 0.
\end{align}
By adding~\eqref{eq:ineq-in-S} for all $\by' \in S$ with~\eqref{eq:ineq-outside-S} for all $\by' \in \Delta_{k, m} \setminus S$, we have
\begin{align}
\sum_{\by \in \Delta_{k, m}} \left(\sum_{\by' \in S} e^{\left<\tau(\by'), \by\right> - \beta(\by' )} - \sum_{\by' \in \Delta_{k, m} \setminus S} e^{\left<\tau(\by'), \by\right> - \beta(\by' )  - \eps} \right) p_{\by} &\nonumber \\
+ \sum_{\by \in \Delta_{k, m}} \left(\sum_{\by' \in \Delta_{k, m} \setminus S} e^{\left<\tau(\by'), \by\right> - \beta(\by' )} - \sum_{\by' \in S} e^{\left<\tau(\by'), \by\right> - \beta(\by' )  - \eps} \right) p'_{\by} &\geq 0. \label{eq:non-negative-comb}
\end{align}

Now, for all $\by \in S$, we can upper bound the coefficient of $p'_{\by}$ in~\eqref{eq:non-negative-comb} by
\begin{align*}
&\sum_{\by' \in \Delta_{k, m} \setminus S} e^{\left<\tau(\by'), \by\right> - \beta(\by' ) } - \sum_{\by' \in S} e^{\left<\tau(\by'), \by\right> - \beta(\by' ) - \eps} \\
&\leq \sum_{\by' \in \Delta_{k, m} \setminus \{\by\}} e^{\left<\tau(\by'), \by\right> - \beta(\by' ) } - e^{\left<\tau(\by), \by\right> - \beta(\by) - \eps} \\
&\overset{\eqref{eq:coeff-dominate}}{\leq} -0.5 e^{\left<\tau(\by), \by\right> - \beta(\by) - \eps} \\
&\overset{\eqref{eq:range-coeff-exponent}}{\leq} -e^{\zeta-1-\eps}.
\end{align*}
Similarly, for all $\by \in \Delta_{k, m} \setminus S$, the coefficient of $p_{\by}$ in~\eqref{eq:non-negative-comb} is at most $- e^{\zeta-1-\eps}$.

Moreover, for all $\by \in S$, we can upper bound the coefficient in~\eqref{eq:non-negative-comb} of $p_{\by}$ by
\begin{align*}
\sum_{\by' \in S} e^{\left<\tau(\by'), \by\right> - \beta(\by' )} - \sum_{\by' \in \Delta_{k, m} \setminus S} e^{\left<\tau(\by'), \by\right> - \beta(\by' )  - \eps} &\leq \sum_{\by' \in \Delta_{k, m}} e^{\left<\tau(\by'), \by\right> - \beta(\by' ) } \\
&\overset{\eqref{eq:coeff-dominate}}{\leq} \left(1 + \frac{1}{2e^{\eps}}\right) e^{\left<\tau(\by), \by\right> - \beta(\by)} \\
&\overset{\eqref{eq:range-coeff-exponent}}{\leq} 2.
\end{align*}
Similarly, for all $\by \in S$, the coefficient of $p'_{\by}$ in~\eqref{eq:non-negative-comb} is at most 2.

Plugging these back into~\eqref{eq:non-negative-comb}, we have
\begin{align*}
0 \leq 2 \left(\sum_{\by \in S} p_{\by} + \sum_{\by \in \Delta_{k, m} \setminus S} p'_{\by}\right) - e^{\zeta - 1 - \eps} \left(\sum_{\by \in S} p'_{\by} + \sum_{\by \in \Delta_{k, m} \setminus S} p_{\by}\right).
\end{align*}
Now, using the fact that $\sum_{\by \in \Delta_{k, m}} p_{\by} = \sum_{\by \in \Delta_{k, m}} p'_{\by} = 1$, we can further simplify the above to
\begin{align*}
2 e^{\zeta - 1 - \eps} &\leq (2 + e^{\zeta - 1 - \eps})\left(\sum_{\by \in S} p_{\by} + \sum_{\by \in \Delta_{k, m} \setminus S} p'_{\by}\right).
\end{align*}
This means that
\begin{align*}
\left(\sum_{\by \in S} p_{\by} + \sum_{\by \in \Delta_{k, m} \setminus S} p'_{\by}\right) \geq \frac{2 e^{\zeta - 1 - \eps}}{2 + e^{\zeta - 1 - \eps}} \geq 2^{-O_{\eps}(m^2 \log k)},
\end{align*}
where the second inequality follows from~\eqref{eq:range-coeff-exponent}.
This establishes~\eqref{eq:obj-small} and hence we have $\SD(\bY, \bY') \leq 1 - 2^{-O_{\eps}(m^2 \log k)}$ as desired.
\end{proof}

\subsection{Limitations of the Lower Bound Approach}\label{sec:MGF_limit}

In this subsection, we argue that the bound we achieve in Lemma~\ref{lem:bounded-moment-ratio-to-sd} is essentially tight, even for $k = 2$. In other words, our approach of using only bounded MGF ratio property and the total variation distance bound from Theorem~\ref{thm:chan-local} cannot give any lower bound better than $O_{\eps}(\sqrt{\log n})$. Specifically, the main lemma of this section is stated below. 

\begin{lemma} \label{lem:gaussian-mgf-bounded-two-dim}
For every $\epsilon > 0$ and $\gamma \in (0, 0.5)$, there exist two random variables $\bY, \bY'$ supported on (subsets of) $\Delta_{2, m}$ for some $m = O_{\eps}(\sqrt{\log(1/\gamma)})$ such that $\SD(\bY, \bY') \geq 1 - \gamma$ and that $\bY, \bY'$ satisfy the $e^{\eps}$-bounded MGF ratio property.
\end{lemma}

Similar to when we analyze our binary summation protocol in Section~\ref{sec:bit-sum-protocol}, it will be more convenient to consider the one-dimensional case, where the two random variables are from $\{0, 1, \dots, m\}$ rather than $\Delta_{2, m}$. In other words, it is more convenient to state our result in this section as follows:

\begin{lemma} \label{lem:gaussian-mgf-bounded}
For every $\epsilon > 0$ and $\gamma \in (0, 0.5)$, there exist two random variables $Y^0$ and $Y^1$ supported on $\{0, \dots, m\}$ for some $m = O_{\epsilon}(\sqrt{\log(1/\gamma)})$ such that $\SD(Y^0, Y^1) \geq 1 - \gamma$ and that $Y^0, Y^1$ satisfy $e^{\eps}$-bounded MGF ratio property.
\end{lemma}

Similar to the analogous statement in Section~\ref{sec:bit-sum-protocol}, it is easy to see that Lemma~\ref{lem:gaussian-mgf-bounded} implies Lemma~\ref{lem:gaussian-mgf-bounded-two-dim}.

\begin{proof}[Proof of Lemma~\ref{lem:gaussian-mgf-bounded-two-dim} from Lemma~\ref{lem:gaussian-mgf-bounded}] For any $\eps > 0$ and $\gamma \in (0, 0.5)$, let $Y^0, Y^1$ be the random variables from Lemma~\ref{lem:gaussian-mgf-bounded} whose values are from $\{0, 1, \dots, m\}$ where $m = O_{\eps}(\sqrt{\log(1/\gamma)})$. We define the random variable $\bY, \bY'$ by $\bY = (Y^0, m - Y^0)$ and $\bY' = (Y^1, m - Y^1)$. Clearly, $\SD(\bY, \bY') = \SD(Y^0, Y^1) \geq 1 - \gamma$. Finally, for any $\bt = (t_1, t_2) \in \R^2$ we have
\begin{align*}
\frac{\bM_{\bY}(\bt)}{\bM_{\bY'}(\bt)} = \frac{\bM_{Y^0}(t_1 - t_2)}{\bM_{Y^1}(t_1 - t_2)},
\end{align*}
which lies in $[e^{-\eps}, e^{\eps}]$ due to the $e^{\eps}$-bounded MGF ratio property of $Y^0, Y^1$.
\end{proof}

\subsubsection{Discrete Gaussian Distributions}

Our construction for Lemma~\ref{lem:gaussian-mgf-bounded} will be based on the discrete Gaussian distribution, which we define below. To do so, we start by defining the (one-dimensional) Gaussian function centered at $c$ with parameter $s$ as
\begin{align*}
\rho_{s,c}(x) = \exp\left(-\frac{\pi (x - c)^2}{s^2}\right),
\end{align*}
for all $x \in \R$. For any countable set $A \subseteq \R$, we define $\rho_{s, c}(A)$ as $\sum_{x \in A} \rho_{s, c}(x)$. For any countable set $A \subseteq \R$ such that $\sum_{x \in A} \rho_{s, c}(x)$ is finite, we may define the discrete
Gaussian distribution over $A$ centered at $c$ with parameter $s$ denoted by $\cD_{A, s, c}$ by
\begin{align*}
\cD_{A, s, c}(x) = \frac{\rho_{s, c}(x)}{\rho_{s, c}(A)},
\end{align*}
for all $x \in A$. Throughout this work, we only use $A$ that is either finite or an additive subgroup of $\Z$; for both cases, it is not hard to see that $\rho_{s, c}(A)$ is finite and hence we will not state this condition again. For brevity, we sometimes drop the subscript $c$ when $c = 0$.

We will use a well-known property of lattices (cf.~\cite{MicciancioR07,GentryPV08,AgrawalGHS13}). Since we will be using this property only in one dimension, we shall not fully define the notion of lattices for higher dimensions. Recall that a one-dimensional lattice is an additive subgroup $a\Z := \{a t \mid t \in \Z\}$ for some $a \in \R^+$.
%
Informally speaking, the property we use is that, if we choose $s$ to be sufficiently large, ``shifting'' the discrete Gaussian distribution by $c$ does not change its normalization factor too much. This is stated more formally below. (For reference, please refer to~\cite[Lemma 2.6]{GentryPV08} which states a more general version of the statement that also works for higher-dimensional lattices.)

\begin{lemma} \label{lem:discrete-gaussian-shift}
For any constants $a, \delta \in \R^+$, there exists a sufficiently large constant $s^* = s^*(a, \delta)$ such that, for any $c \in \R$, the following holds:
\begin{align} \label{eq:ratio-gaussian-shift}
\frac{\rho_{s^*, c}(a\Z)}{\rho_{s^*}(a\Z)} \in [e^{-\delta}, 1]. 
\end{align}
\end{lemma}

We will also use the following observation that, similar to the (continuous) Gaussian distribution, we may choose a sufficiently large truncation point $\ell^* a$ for which the total mass of all points $x$ with $|X - c| > \ell^*a$ is arbitrarily small. Note that the only reason the observation is not completely trivial is that the truncation point should work for all centers $c$. Nonetheless, the proof of the observation is still rather straightforward, and we defer it to Appendix~\ref{app:discrete-gaussian-tail}.

\begin{observation} \label{obs:discrete-gaussian-tail}
For any constants $a, \delta \in \R^+$, let $s^* = s^*(a, \delta)$ be as in Lemma~\ref{lem:discrete-gaussian-shift}. Then, for any $\lambda > 0$, there exists a sufficiently large positive integer $\ell^* = \ell^*(a, \delta, \lambda)$ such that, for any $c \in \R$, we have
\begin{align*}
\Pr_{X \sim \cD_{a\Z, s^*, c}}[|X - c| > \ell^* a] \leq \lambda.
\end{align*}
\end{observation}

\subsubsection{Proof of Lemma~\ref{lem:gaussian-mgf-bounded}}

Having stated the necessary background, we now describe our construction, starting with an informal intuition; all arguments will be subsequently formalized. Distributions of both $Y^0, Y^1$ will place $\frac{\gamma}{2}$ probability masses at each of $0$ and $m$, and these two points shared by the supports of $Y^0$ and $Y^1$. (This ensures that the total variation distance of $Y^0$ and $Y^1$ are at least $1 - \gamma$.) In the middle, we then place discrete Gaussian distributions centered at $c = m/2$ for $Y^0$ and $Y^1$, with that of $Y^0$ only supported on even numbers whereas that of $Y^1$ supported on odd numbers. These discrete Gaussian distributions are truncated so that the supports are within the range of $[c - w, c + w]$ for some parameter $w$.

The reason behind the construction is as follows. First, when $|t| \geq O_{\eps}(\sqrt{\log(1/\gamma)})$, it is not hard to see that the MGFs at $t$ are dominated by the terms corresponding to the points 0 or $m$. Our parameters are selected in such a way that, when this is not the case, it must be that $|t| \ll w$. In this case, we observe that the MGFs of discrete Guassian distributions are simply proportional to normalization terms of other discrete Gaussian distributions, shifted by $O(t)$ (and truncated appropriately). (See~\eqref{eq:shift-gaussian-moment} below.) Since $|t| \ll w$, we can then apply Lemma~\ref{lem:discrete-gaussian-shift} and Observation~\ref{obs:discrete-gaussian-tail} to get a good bound on these terms. This concludes the main ideas in the proof, which is presented more formally below.

\begin{proof}[Proof of Lemma~\ref{lem:gaussian-mgf-bounded}]
We will assume w.l.o.g. that $\eps \leq 0.1$, as otherwise we may consider the case $\eps = 0.1$ instead. Before we can describe and analyze the distributions, we have to specify certain parameters:
\begin{itemize}
\item Let $s = s^*(2, \eps/4)$ from Lemma~\ref{lem:discrete-gaussian-shift} (i.e., for $2\Z$ lattice and $\delta = \eps/4$).
\item Let $\ell = \ell^*(2, \eps/4, 1 - e^{-\eps/4})$ from Observation~\ref{obs:discrete-gaussian-tail} (i.e., for $2\Z$ lattice, $\delta = \eps/4$ and $\lambda = 1 - e^{-\eps/4}$).
\item Let $w = \frac{s^2 \sqrt{\log(1/\gamma)}}{\pi} + 2\ell^*$, $c = \left\lceil w + \log\left(\frac{2}{(e^{\eps} - 1)}\right) + \sqrt{\log(1/\gamma)}\right\rceil$ and $m = 2c$.
\end{itemize}

Let $S_0$ denote the set $2\Z \cap [c - w, c + w]$ and $S_1$ denote $(2\Z + 1) \cap [c - w, c + w]$.
Let $\mu$ be the distribution that has probability mass $0.5$ at 0 and $0.5$ at $m$. We let $Y^0$ be sampled from the mixture distribution $\gamma \cdot \mu + (1 - \gamma) \cdot \cD_{S_0, s, c}$ and $Y^1$ be sampled from the mixture distribution $\gamma \cdot \mu + (1 - \gamma) \cdot \cD_{S_1, s, c}$. Figure~\ref{fig:truncated_gaussian} illustrates an example of the two distributions.

\begin{figure}
    \centering
    \includegraphics[width=\textwidth]{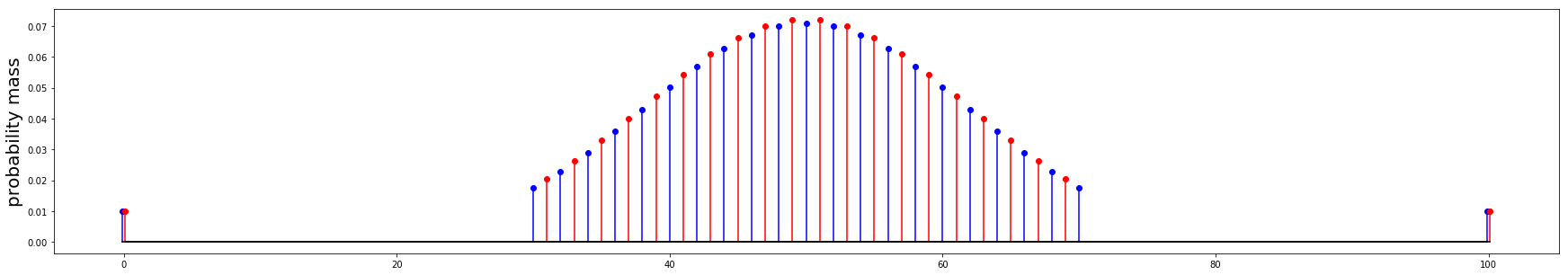}
    \caption{The probability mass functions of $\gamma \cdot \mu + (1 - \gamma) \cdot \cD_{S_0, s, c}$ and $\gamma \cdot \mu + (1 - \gamma) \cdot \cD_{S_1, s, c}$ for parameters $\gamma = 0.02, c = 50, m = 2c = 100, w = 20, s = 30$. The $x$-axis corresponds to the value of the random variable and the $y$-axis corresponds to the probability mass at that value. The red points and the blue points correspond respectively to $\gamma \cdot \mu + (1 - \gamma) \cdot \cD_{S_0, s, c}$ and $\gamma \cdot \mu + (1 - \gamma) \cdot \cD_{S_1, s, c}$. 
    }
    \label{fig:truncated_gaussian}
\end{figure}

Observe that $\supp(Y^0) \cap \supp(Y^1) = \{0, m\}$, and each of the two points has mass $\gamma / 2$. Hence, we have $\SD(Y^0, Y^1) = 1 - \gamma$ as desired.

We will next verify that $Y^0, Y^1$ satisfies $e^{\eps}$-bounded MGF ratio. To do this, observe that for $i \in \{0, 1\}$, 
\begin{align} \label{eq:mgf-linear-comb}
\bM_{Y^i}(t) = \gamma \cdot \bM_{\mu}(t) + (1 - \gamma) \cdot \bM_{\cD_{S_i, s, c}}(t).
\end{align}

We now consider two cases, based on whether $|t| > \sqrt{\log(1/\gamma)}$.

\begin{enumerate}
\item $|t| > \sqrt{\log(1/\gamma)}$. There are two subcases here: $t > \sqrt{\log(1/\gamma)}$ or $t < -\sqrt{\log(1/\gamma)}$. Let us first assume that $t > \sqrt{\log(1/\gamma)}$. In this case, since the maximum number in $\supp(\cD_{S_i, s, c})$ is at most $c + w$, we have $\bM_{\cD_{S_i, s, c}}(t) \leq e^{t(c + w)}$. On the other hand, we have $\bM_{\mu}(t) \geq \frac{e^{t m}}{2} = \frac{e^{2tc}}{2}$. Hence, we have
\begin{align*}
\frac{\bM_{\cD_{S_i, s, c}}(t)}{\bM_{\mu}(t)} \leq 2e^{t(-c + w)} \leq (e^{\eps} - 1)\gamma, 
\end{align*}
where the inequality comes from our choice of $c$.

As a result, from~\eqref{eq:mgf-linear-comb}, we have
\begin{align*}
\gamma \cdot \bM_{\mu}(t) \leq \bM_{Y^i}(t) = \gamma \cdot \bM_{\mu}(t) + (1 - \gamma) \cdot \bM_{\cD_{S_i, s, c}}(t) \leq e^{\eps} \gamma \cdot \bM_{\mu}(t).
\end{align*}
Thus, $\frac{\bM_{Y^1}(t)}{\bM_{Y^2}(t)} \in [e^{-\eps}, e^{\eps}]$ as desired. The subcase $t < -\sqrt{\log(1/\gamma)}$ is similar; in particular, we also have $\frac{\bM_{\cD_{S_i, s, c}}(t)}{\bM_{\mu}(t)} \leq \frac{e^{t(c - w)}}{0.5} = 2e^{t(c - w)} \leq (e^{\eps} - 1)\gamma$, which results in the same conclusion.
\item $|t| \leq \sqrt{\log(1/\gamma)}$. In this case, we further rearrange $\bM_{\cD_{S_i, s, c}}(t)$ as
\begin{align}
\bM_{\cD_{S_i, s, c}}(t) &= \sum_{y \in S_i} \cD_{S_i, s, c}(y) \cdot e^{ty} \nonumber \\
&= \sum_{y \in S_i} \frac{\rho_{s, c}(y)}{\rho_{s, c}(S_i)} \cdot e^{ty} \nonumber \\
&= \frac{1}{\rho_{s, c}(S_i)} \sum_{y \in S_i} e^{-\frac{\pi(y - c)^2}{s^2} + ty} \nonumber \\
&= \frac{1}{\rho_{s, c}(S_i)} \sum_{y \in S_i} e^{-\frac{\pi(y - c - 0.5s^2t / \pi)^2}{s^2} + \frac{\pi((c + 0.5s^2t / \pi)^2 - c^2)}{s^2}} \nonumber \\
&= \frac{1}{\rho_{s, c}(S_i)} \sum_{y \in S_i} e^{-\frac{\pi(y - c - 0.5s^2t / \pi)^2}{s^2} + 0.5t(2c + 0.5s^2t / \pi)} \nonumber \\
&= \frac{e^{0.5t(2c + 0.5s^2t / \pi)}}{\rho_{s, c}(S_i)} \sum_{y \in S_i} e^{-\frac{\pi(y - c - 0.5s^2t / \pi)^2}{s^2}} \nonumber \\
&= e^{0.5t(2c + 0.5s^2t / \pi)} \cdot \frac{\rho_{s, c + 0.5s^2 t/\pi}(S_i)}{\rho_{s, c}(S_i)}. \label{eq:shift-gaussian-moment}
\end{align}
Now, observe that
\begin{align*}
\rho_{s, c}(2\Z + i) &\geq \rho_{s, c}(S_i) \\
&= \rho_{s, c}(2\Z + i) \cdot \left(1 - \Pr_{X \sim \cD_{2\Z + i, s, c}}[X \notin S_i]\right) \\
&\geq \rho_{s, c}(2\Z + i) \cdot \left(1 - \Pr_{X \sim \cD_{2\Z + i, s, c}}[|X - c| \geq w]\right) \\
(\text{from our choice of } w) &\geq \rho_{s, c}(2\Z + i) \cdot \left(1 - \Pr_{X \sim \cD_{2\Z + i, s, c}}[|X - c| \geq 2\ell^*]\right) \\
&= \rho_{s, c}(2\Z + i) \cdot \left(1 - \Pr_{X \sim \cD_{2\Z, s, c - i}}[|X - c - i| \geq 2\ell^*]\right) \\
&\geq e^{-\eps/8} \cdot \rho_{s, c}(2\Z + i),
\end{align*}
where the last inequality comes from our choice of $\ell^*$.

Similarly, observe that
\begin{align*}
\rho_{s, c + 0.5s^2 t/\pi}(2\Z + i) &\geq \rho_{s, c + 0.5s^2 t/\pi}(S_i) \\
&= \rho_{s, c + 0.5s^2 t/\pi}(2\Z + i) \cdot \left(1 - \Pr_{X \sim \cD_{2\Z + i, s, c + 0.5s^2 t/\pi}}[X \notin S_i]\right) \\
&\geq \rho_{s, c + 0.5s^2 t/\pi}(2\Z + i) \cdot \left(1 - \Pr_{X \sim \cD_{2\Z + i, s, c + 0.5s^2 t/\pi}}[|X - c| \geq w - |0.5s^2 t/\pi|]\right) \\
(\text{from our choice of } w) &\geq \rho_{s, c + 0.5s^2 t/\pi}(2\Z + i) \cdot \left(1 - \Pr_{X \sim \cD_{2\Z + i, s, c + 0.5s^2 t/\pi}}[|X - c| \geq 2\ell^*]\right) \\
&= \rho_{s, c + 0.5s^2 t/\pi}(2\Z + i) \cdot \left(1 - \Pr_{X \sim \cD_{2\Z, s, c + 0.5s^2 t/\pi - i}}[|X - c - i| \geq 2\ell^*]\right) \\
&\geq e^{-\eps/8} \cdot \rho_{s, c + 0.5s^2 t/\pi}(2\Z + i).
\end{align*}

Plugging the above two inequalities back into~\eqref{eq:shift-gaussian-moment}, we have
\begin{align}
\bM_{\cD_{S_i, c, s}}(t) &\in \left[e^{-\eps/8}, e^{\eps/8} \right] \cdot e^{0.5t(2c + 0.5s^2t / \pi)} \cdot \frac{\rho_{s, c + 0.5s^2 t/\pi}(2\Z + i)}{\rho_{s, c}(2\Z + i)} \nonumber \\
&= \left[e^{-\eps/8}, e^{\eps/8} \right] \cdot e^{0.5t(2c + 0.5s^2t / \pi)} \cdot \frac{\rho_{s, c + 0.5s^2 t/\pi - i}(2\Z)}{\rho_{s, c - i}(2\Z)}. \label{eq:gaussian-moment-intermediate}
\end{align}
Finally, from our choice of $s$, we have that $\frac{\rho_{s, c - i}(2\Z)}{\rho_{s}(2\Z)}, \frac{\rho_{s, c + 0.5s^2 t/\pi - i}(2\Z)}{\rho_{s}(2\Z)} \in [e^{-\eps/4}, 1]$. Combining these with~\eqref{eq:gaussian-moment-intermediate}, we have
\begin{align*}
\bM_{\cD_{S_i, c, s}}(t) &\in \left[e^{-\eps/2}, e^{\eps/2}\right] \cdot e^{0.5t(2c + 0.5s^2t / \pi)}.
\end{align*}
As a result, we must have $\frac{\bM_{\cD_{S_1, c, s}}(t)}{\bM_{\cD_{S_2, c, s}}(t)} \in [e^{-\eps}, e^{\eps}]$.
From this and from~\eqref{eq:mgf-linear-comb}, we have $\frac{\bM_{Y^1}(t)}{\bM_{Y^2}(t)} \in [e^{-\eps}, e^{\eps}]$ as desired. \qedhere
\end{enumerate}
\end{proof}

\section{From Binary Summation to Real Summation}\label{sec:real-summation}

In this section use our  pure-$\shuffledDP$ protocol for binary summation in Section~\ref{sec:bit-sum-protocol} to obtain a pure-$\shuffledDP$ protocol for  summation of real numbers in the interval $[0, 1]$.  More precisely we show the following, which is a more quantitative version of Theorem~\ref{th:pure_real_agg_protocol_informal}.

\begin{theorem}\label{th:pure_real_agg_protocol}
	For every sufficiently large $n$ and $\eps \in (0, 1)$ there is an $\epsilon$-$\shuffledDP$ protocol for summation for inputs $x_1, \ldots, x_n\in [0,1]$, where each user sends $O \left(\frac{\log^3 n}{\ep}\right)$ messages each of length $O(\log \log n)$ bits to the analyzer, and has expected error at most $O\left(\frac{\sqrt{\log (1/\ep)}}{\ep^{3/2}}\right)$.
\end{theorem}


The randomizer and analyzer of the protocol are shown as Algorithms~\ref{alg:real-randomizer} and~\ref{alg:real-analyzer} respectively\footnote{Note that $x_{i}[j]$ denote the $j$th bit in a binary representation of $x\in [0, 1]$, such that $x_i = \sum_{j=1}^\infty x[j]/2^j$ (e.g., the representation of $x=1$ has $x[j]=1$ for $j = 1,2,\dots$).} (the sequence $(\eps_j)_{j \in \N}$ will be specified below.). The idea is to round each input to $2\log n$ bits of precision (resulting in a negligible rounding error) and then run an independent binary summation protocol for each bit position.
By ``attaching'' $j$ to each message from the binary summation protocol for bit position $j$, we can run all protocols as a single shuffle, using composition to bound the total privacy loss.
(We observe that composition of $t$ independent shuffled model protocols into a single protocol is possible in general, at the expense of increasing the number of bits in each message by $\log t$.)
By allocating a large share of the privacy budget to the most significant bits, the error can be kept within a constant factor of the error for binary summation.
The communication complexity is somewhat larger than that of the binary summation protocol: 
the number of messages per user is increased by roughly a factor of $O(\log^2 n)$ and each message is  about $\log\log n$ bits (since we need $2\log n$ different symbols).

\begin{proof}[Proof of Theorem~\ref{th:pure_real_agg_protocol}]
For each $j = 1, \dots, \lceil 2 \log n \rceil$, we let $\eps_j = \max\{\frac{0.9^j \eps}{20}, \frac{\eps}{4\log n}\}$. The multiset of all messages output by {\sc RealRandomizer}$_{(\eps_j)_{j \in \N}, n}(x_i)$, for $i=1,\dots,n$ is in one-to-one correspondence with the sequence of multisets output by {\sc BinaryRandomizer}$_{\eps_j, n}(x_i[j]))$, for $j=1,\dots,2\log n$. Thus, we can use composition (see, e.g.,~\cite[Theorem~3.15]{DworkRothBook}) to bound the privacy parameter of the combined protocol by the sum of privacy parameters $\eps_j$:

$$\sum_{j=1}^{2\log n} \eps_j 
\leq \sum_{j=1}^{2\log n} \left(\frac{0.9^j \eps}{20} + \frac{\eps}{4 \log n}\right)
\leq \eps \enspace .$$

Hence, the protocol is $\eps$-DP. Next, we consider the expected error of the analyzer. 
Let $\bar{x}_i = \sum_{j=1}^{2\log n} x_i[j]$ be the rounded version of $x_i$. Since $|\sum_{i=1}^n \bar{x}_i - \sum_{i=1}^n x_i| \leq \sum_{i=1}^n |\bar{x}_i - x_i| < 1/n$, it suffices to argue that the protocol outputs a good approximation of $\sum_{i=1}^n \bar{x}_i$.
To do so, let $j^*$ be the smallest integer for which $\eps_{j^*} = \frac{\eps}{4 \log n}$. 
Recall from Theorem~\ref{th:pure_bin_agg_protocol} that the expected error from the $j$th bit analyzer is at most $O\left(\frac{\sqrt{\log(1/\eps)}}{\eps^{3/2}}\right)$. Since the real summation analyzer outputs a weighted sum of contributions for each bit position obtained from the binary sum analyzers, the total error in the weighted sum returned by the analyzer is bounded by
\begin{align*}
O\left(\sum_{j=1}^{2\log n}\frac{1}{2^j} \cdot \frac{\sqrt{\log(1/\eps_j)}}{\eps_j^{3/2}}\right) &= O\left(\sum_{j=1}^{j^* - 1} \frac{1}{2^j} \cdot \frac{\sqrt{\log(1/\eps) + j}}{(0.81)^{1.5j} \cdot \eps^{3/2}} + \sum_{j=j^*}^{2 \log n} \frac{1}{2^j} \cdot \frac{\sqrt{\log n} \cdot \sqrt{\log(1/\eps) + \log\log n}}{\eps^{3/2}}\right) \\
&\leq O\left(\frac{\sqrt{\log(1/\eps)}}{\eps^{3/2}}\right) + O\left(\frac{\sqrt{\log(1/\eps)}}{\eps^{3/2}} \cdot \frac{(\log n)^{3/2} \cdot \sqrt{\log \log n}}{2^{j^*}}\right) \\
&\leq O\left(\frac{\sqrt{\log(1/\eps)}}{\eps^3/2}\right),
\end{align*}
where the last inequality follows from our choice of $j^*$, which by definition of $\eps_j$ implies that $0.9^{j^*} \leq \frac{5}{\log n}$.

Finally, we consider the number of messages $\sum_{j=1}^{2\log n} d_j$ sent by each randomizer. From Theorem~\ref{th:pure_bin_agg_protocol}, we have $d_j = O\left(\frac{\log n}{\eps_j}\right)$. Hence, the total number of messages sent per user is
\begin{align*}
\sum_{j=1}^{2\log n} d_j = O\left(\sum_{j=1}^{\log n} \frac{\log n}{\eps_j}\right) \leq O\left(\sum_{j=1}^{\log n} \frac{\log n}{\eps/(4 \log n)}\right) = O\left(\frac{\log^3 n}{\eps}\right),
\end{align*}
which completes our proof.
\end{proof}

\section{Conclusion and Open Questions}\label{sec:conc_oqs}

In this work, we gave the first pure-$\shuffledDP$ protocols for binary and real summation with constant error. We further prove a communication lower bound for any non-interactive protocols for binary summation. While these have advanced our understanding of pure-$\shuffledDP$ protocols, there are still many questions left open after this work. Specifically, the immediate open questions are:
\begin{itemize}[nosep]
\item Can we improve the error guarantee in the (binary and real) summation protocols to achieve the asymptotically optimal guarantee of $1/\eps$, which can be achieved by $\centralDP$ protocols~\cite{dwork2006calibrating}?
\item What is the optimal per user communication complexity of non-interactive $\shuffledDP$ protocols for binary and real summation? As we have shown, the communication complexity for binary summation lies between $O_{\eps}(\log n)$ and $\Omega_{\eps}(\sqrt{\log n})$. On the other hand, for real summation, the only lower bound is the trivial $\Omega(\log n)$ bound (which holds even without privacy concerns) whereas our upper bound is $O_{\eps}(\log^3 n)$. We remark here that, our approach for real summation (of running the pure-DP binary summation protocol independently for each coordinate in the base-2 representation) cannot achieve better than $O_{\eps}(\log^{3/2} n)$ communication complexity, because we have to consider $\Omega_{\eps}(\log n)$ coordinates and, from our lower bound, each coordinate requires at least $\Omega_{\eps}(\sqrt{\log n})$ bits of communication. 
\item
In Appendix~\ref{sec:app_hist}, we show that our binary summation protocol also yields a pure-DP protocol for histograms (aka frequency estimator) with error $O_{\eps}(\log B \log n)$ but with \emph{linear} per user communication complexity. The latter is in contrast to the approximate-DP multi-message protocol of~\cite{anon-power}, which has a per user communication complexity of only $O_{\eps}(\poly(\log n, \log B))$ and incurs a similar error of $O_{\eps}(\poly(\log n, \log B))$ bits. It is hence a very interesting open question to come up with (or rule out) a pure-DP protocol with a smaller communication complexity.
\item Can we exploit interactivity to break our $\Omega_{\eps}(\sqrt{\log n})$ communication lower bound? Alternately, can we prove any non-trivial lower bound that holds also with interaction?
\end{itemize}

On a high-level, it would also be interesting to develop tools to help prove guarantees for pure-$\shuffledDP$ protocols. In the case of approximate-DP, there are amplification theorems~\cite{erlingsson2019amplification,BalleBGN19} that can yield an approximate-$\shuffledDP$ protocol from a $\localDP$ protocol.  Although this may not be optimal in some cases (as shown by the multi-message protocols in~\cite{anon-power,balle_privacy_2019constantIKOS,ghazi2019private}), such theorems can be conveniently applied to a large class of protocols and yield good approximate-DP guarantees. On the other hand, our proofs in this work are specific to our carefully designed protocols. It would be much more convenient if one can give a unifying theorem that proves pure privacy guarantees for any protocol with easily verifiable conditions. 


\section*{Acknowledgements}
We are grateful to Borja Balle, Kunal Talwar, and Vitaly Feldman for helpful discussions.

\bibliographystyle{alpha}
\bibliography{main.bib}

\appendix

\section{Pure Protocol for Histograms}\label{sec:app_hist}
A well-studied generalization of binary summation is the problem of computing \emph{histograms} (aka \emph{point functions} or \emph{frequency estimation}), where each of $n$ users is given an element in the set $\{1,\dots, B\}$ and the goal is to estimate the number of users holding any element $j \in \{1,\dots, B\}$, and with the smallest possible $\ell_{\infty}$ error (across the $B$ coordinates). For $B = 2$, this reduces to binary summation.

The smallest possible error for computing histograms is $\Theta(\min(\log{B}, \log(1/\delta))/\epsilon)$~\cite{dwork2006calibrating,bun2016simultaneous,bassily2015local,hardt2009geometry} in the central model and $\Theta(\sqrt{n \log{B}}/\epsilon)$~\cite{bassily2015local} in the local model. Recent work of~\cite{anon-power} gave an approximate-DP protocol with error $O\left(\log{B} + \frac{\sqrt{\log{B} \log(1/(\epsilon \delta))}}{\epsilon}\right)$ where each user sends $O\left(\frac{\log(1/(\epsilon \delta))}{\epsilon^2}\right)$ messages (each consisting of $O(\log{B} \log{n}))$ bits), and the subsequent work of~\cite{balcer2019separating} gave an approximate-DP protocol in the multi-message shuffled model with an (incomparable) error of $O(\log(1/\delta)/\epsilon^2)$ but with each user communicating a very large number $O(B)$ of messages.

Our pure binary summation protocol (Theorem~\ref{th:pure_bin_agg_protocol_informal}) implies as a black-box the first pure-DP protocol with polylogarithmic error for computing histograms, albeit with very large communication.

\begin{corollary}\label{cor:histograms}
For every positive real number $\eps$, there is an $\eps$-$\shuffledDP$ protocol that computes histograms on domains of size $B$ with an expected $\ell_{\infty}$ error of at most $O_{\epsilon}(\log{B} \log{n})$, and where each user sends $O_{\epsilon}(B\log{n})$ messages each consisting of $O(\log{B})$ bits.
\end{corollary}

\if 0\pasin{More formal proof of this corollary? Also, better accuracy?}
\noah{
  We attempt to argue as follows: \todo{using the same notation as previously}, we have, for any $C > 0$,
    \begin{align}
    \Pr\left[ \left|Y - \sum_{i=1}^n x_i \right| > C\right]
&\leq 
      \sum_{m=0}^n\Pr[M = m] \cdot \left( \Pr\left[\left|Z_1 + \cdots + Z_m - \frac{md}{2} \right| > C - \frac m2 \right]\right) \nonumber\\
    \text{(Since $Z_i \in [0,d]$) }     & \leq \exp(-\Omega(C)) + \sum_{m=2C/d}^{C/100}\Pr[M = m] \cdot \left( \Pr\left[\left|Z_1 + \cdots + Z_m - \frac{md}{2} \right| > \frac C2 \right]\right) \nonumber
    \end{align}
    But the random variable $Z_1 + \cdots + Z_m - \frac{md}{2}$ is sub-Gaussian with variance proxy $\Theta(d^2 \cdot m)$. Thus the best bound we can get on $ \Pr\left[\left|Z_1 + \cdots + Z_m - \frac{md}{2} \right| > \frac C2 \right]$ is an upper bound $\exp(-\Theta(\frac{C}{d^2 m}))$. When $C < md$, this doesn't seem to give anything...
  }\noah{Update: the bound we have now is fine, but perhaps can do better using Bernstein's inequality. I will confirm either way tmrw.}
  \fi
The proof of Corollary~\ref{cor:histograms} is very simple: we just run our $(\eps/2)$-DP binary summation protocol for each coordinate $j \in B$ independently and attach to the message the coordinate index $j$ (similar to our real summation protocol). It is obvious to see that the number of messages and the message length are as claimed. The $\ell_{\infty}$ error bound can be seen as follows. We claim that the probability that the $\ell_{\infty}$ error is more than $C d \log B = O_{\eps}(C \log B \log n)$ for any sufficiently large $C$ is at most $\exp(-\Omega_{\eps}(C))$; this would immediately imply the desired expected $\ell_{\infty}$ error bound stated in Corollary~\ref{cor:histograms}. 

Now, to see that the probabilistic statement above is true, we first consider each coordinate separately. Since each user picks from the ``noise distribution'' for this coordinate with probability $p \leq O_{\eps}(1/n)$, a standard application of the Chernoff bound implies that the probability that the number of users picking from the noise distribution for this coordinate exceeds $C \log B$ is at most $\exp(-\Omega_{\eps}(C \log B))$ for any sufficiently large $C$. When this event does not occur, the error for this coordinate is at most $C d \log B$. Taking a union bound over all the coordinates yields the desired result.


We point out that using the Count Min sketch as in~\cite{MelisDC16,anon-power} would allow us to reduce the per user communication in Corollary~\ref{cor:histograms} to $O(n\log B \log{n})$ messages each consisting of $O(\log{n})$ bits, but further reducing the communication down to $O_{\eps}(\poly(\log n, \log B))$ bits remains a very interesting open question.


\section{Missing Proofs from Section~\ref{sec:bit-sum-protocol}}\label{sec:missing_proofs_bit_sum_prot}

\subsection{Proof of Lemma \ref{lem:pia-lb}}
\label{sec:pia-lb}
In this section we prove Lemma \ref{lem:pia-lb}. We first recall some basic facts about unimodal random variables:
\begin{definition}[Unimodal random variables]
  A random variable $Z$ that takes values on $\{ 0, 1, \ldots, D \}$, for some positive integer $D$, is defined to be \emph{unimodal}, if there is some $k \in \{ 0, 1, \ldots, D \}$ so that for $j \leq k$, the function $j \mapsto \Pr[Z=j]$ is non-decreasing in $j$, and for $j \geq k$, the function $j \mapsto \Pr[Z=j]$ is non-increasing in $j$. In such a case, $k$ is said to be the \emph{mode} of the distribution of $Z$.
\end{definition}

\begin{lemma} \label{lem:unimodal}
The distribution of $Z_1 + Z_2 + \cdots + Z_m$, where $Z_1, \ldots, Z_m \sim \nu = \DLap_d(d/2, s)$, is unimodal with mode(s) given by $ \{ \lfloor md/2 \rfloor, \lceil md/2 \rceil \}$. 
\end{lemma}
\begin{proof}
  Unimodality of $Z_1 + \cdots + Z_m$ follows from log-concavity of  $\DLap_d(d/2, s)$ and the fact that log-concave distributions are  strongly unimodal, meaning that convolving with any unimodal distribution results in another unimodal distribution~\cite[Theorem 3]{keilson1971discrete}.
  
  The fact that the mode is $md/2$ if $m$ is even and that both $\lfloor md/2 \rfloor, \lceil md/2 \rceil$ are modes if $m$ is odd follows by symmetry of $\DLap_d(d/2, s)$. 
\end{proof}

\begin{lemma} \label{lem:normalization-val}
For any $\mu \in \R, w, s > 1$, we have
\begin{align*}
C_w(\mu, s) \leq C(\mu, s) \leq \frac{2}{1 - e^{-1/s}}.
\end{align*}
\end{lemma}
The proof of Lemma \ref{lem:normalization-val} is deferred to Section \ref{sec:sequence-separation}.

\begin{proof}[Proof of Lemma \ref{lem:pia-lb}]
  Lemma \ref{lem:sequence-separation} gives
  \begin{equation}
    \label{eq:pa-lb}
 P_{i+a, j + a \left( \frac{d-1}{2} \right)} \geq P_{a, a \left( \frac{d-1}{2} \right)} \cdot P_{i,j},
\end{equation}
so it suffices to find a suitable lower bound on $P_{a, a \left( \frac{d-1}{2} \right)} = \Pr_{Z_1, \ldots, Z_a \sim \nu} \left[Z_1 + \cdots + Z_a = a \left( \frac{d-1}{2} \right)\right]$. To do so, note that for $i \in \{ 1, \ldots, a\}$, $\E[Z_i] = d/2$, and write $Z = Z_1 + \cdots + Z_a$. By the Marcinkiewicz--Zygmund inequality (Theorem~\ref{thm:anti-concen}) and the power mean inequality, we have 
\begin{align}
  \E \left[ |Z - da/2| \right] & \geq \frac{1}{2\sqrt{2}} \E \left[ \sqrt{\sum_{i=1}^a (Z_i - d/2)^2 } \right]\nonumber\\
                               & \geq \frac{1}{2\sqrt{2a}} \E \left[ \sum_{i=1}^a | Z_i - d/2 | \right] \nonumber\\
  \label{eq:exp-lb}
                               & \geq \frac{\sqrt{a}}{10} \cdot s,
\end{align}
The last inequality above follows since for $Z_i \sim \DLap_d(d/2, s)$, 
\begin{align*}
  \E [|Z_i - d/2|] & \geq (s/2) \cdot \Pr[|Z_i - d/2| \geq s/2] \\
                   \text{ (using Lemma \ref{lem:normalization-val}) }\ & \geq (s/2) \cdot \left( 1 - (s/2) \cdot \frac{1}{C_d(d/2,s)} \right)\\
                   \text{ (since $1-e^{-1/s} \leq 1/s$ )} \ \ & \geq (s/2) \cdot \left( 1 - (s/2) \cdot \frac{1-e^{-1/s}}{2} \right) \\
  & \geq 3s/8.
\end{align*}

Furthermore, we have
\begin{align} \label{eq:var-sum-upper-bound}
\E[(Z - da/2)^2] = \sum_{i=1}^a \Var[Z_i] \leq 2as^2.
\end{align}
As a result, we have
\begin{align}
&\E[|Z - da/2| \mid |Z - da/2| \geq s^2] \cdot \Pr[|Z - da/2| \geq s^2] \nonumber \\
&\leq \frac{1}{s^2} \E[(Z - da/2)^2 \mid |Z - da/2| \geq s^2] \cdot \Pr[|Z - da/2| \geq s^2] \nonumber \\
&\overset{~\eqref{eq:var-sum-upper-bound}}{\leq} \frac{1}{s^2} (2as^2) \nonumber \\
&= 2a. \label{eq:abs-dev-bound-large-val}
\end{align}
Using inequality~\eqref{eq:abs-dev-bound-large-val} above, we may upper bound $\E [ | Z - da/2 | ]$ by
\begin{align} \label{eq:exp-ub}
\E [ | Z - da/2 | ] \leq a/2 + \Pr[a/2 \leq |Z - da/2| < s^2] \cdot s^2 + 2a.
\end{align}
Combining (\ref{eq:exp-lb}), (\ref{eq:exp-ub}), and $a \leq s^2/1000$ gives
\begin{align*}
\Pr[a/2 \leq |Z-da/2| < s^2] \geq \frac{s\sqrt{a}/10 - 2.5a}{s^2} \geq \frac{\sqrt{a}}{20s}.
\end{align*}
Finally, unimodality and symmetry of $Z$ (Lemma \ref{lem:unimodal}) gives
$$
P_{a, a\left( \frac{d-1}{2} \right)} = \Pr\left[ Z = \frac{ad}{2} - \frac{a}{2} \right] \geq \frac{\sqrt{a}}{40s^3},
$$
which, combined with (\ref{eq:pa-lb}), completes the proof. 
\end{proof}


\subsection{Proof of Lemma~\ref{lem:normalization-val}}

\begin{proof}[Proof of Lemma~\ref{lem:normalization-val}]
It is obvious to see that $C_w(\mu, s) \leq C(\mu, s)$. To bound the latter, recall that
\begin{align} \label{eq:normalization-two-terms}
C(\mu, s) = \sum_{z = -\infty}^{\infty} e^{-|z - \mu| / s} \leq \sum_{z = -\infty}^{\lfloor \mu \rfloor} e^{-(\mu - z) / s} + \sum_{z  = \lceil \mu \rceil}^{\infty} e^{-(z - \mu) / s}.
\end{align}
Consider the second term on the right hand side. We have
\begin{align*}
\sum_{z  = \lceil \mu \rceil}^{\infty} e^{-(z - \mu) / s} \leq \sum_{i=0}^{\infty} e^{-i/s} = \frac{1}{1 - e^{-1/s}}.
\end{align*}
Similarly, we also have 
\begin{align*}
\sum_{z  = -\infty}^{\lfloor \mu \rfloor} e^{-(\mu - z) / s} \leq \sum_{i=0}^{\infty} e^{-i/s} = \frac{1}{1 - e^{-1/s}}.
\end{align*}
Plugging the above two inequalities into \eqref{eq:normalization-two-terms}, we get the desired bound.
\end{proof}

\subsection{Proof of Lemma~\ref{lem:sequence-separation}}
\label{sec:sequence-separation}
\begin{proof}[Proof of Lemma~\ref{lem:sequence-separation}]
We have 
\begin{align*}
P_{i + i, j + j'} &= \sum_{\ba \in S_{i + i', j + j'}} \nu(\ba). \\
&= \sum_{a_1, \cdots, a_{i + i'} \in \Z \cap [0, d] \atop a_1 + \cdots + a_{i + i'} = j + j'} \nu(a_1) \cdots \nu(a_{i + i'}) \\
&\geq \sum_{a_1, \cdots, a_{i + i'} \in \Z \cap [0, d] \atop a_1 + \cdots + a_i = j \text{ and } a_{i + 1} \cdots + a_{i + i'} = j'} \nu(a_1) \cdots \nu(a_{i + i'}) \\
&= \left(\sum_{a_1, \cdots, a_{i} \in \Z \cap [0, d] \atop a_1 + \cdots + a_i = j} \nu(a_1) \cdots \nu(a_{i})\right)\left(\sum_{a_{i+1}, \cdots, a_{i + i'} \in \Z \cap [0, d] \atop a_{i + 1} \cdots + a_{i + i'} = j'} \nu(a_{i+1}) \cdots \nu(a_{i + i'})\right) \\
&= P_{i, j} \cdot P_{i', j'}. \qedhere
\end{align*}
\end{proof}

\section{Proof of Theorem~\ref{thm:chan-local}}
\label{app:tv-vs-utility-proof}

In this section, we provide a self-contained proof of Theorem~\ref{thm:chan-local}. Our proof use the following well-known theorem, which provides an anti-concentration guarantee of a sum of independent random variables.

\begin{theorem}[Marcinkiewicz-–Zygmund inequality] \label{thm:anti-concen} 
Let $\xi_1, \dots, \xi_n$ be any independent random variables with mean zero and $\E\left[|\xi_i|\right] < \infty$. Then,
\begin{align*}
\E\left[\left|\sum_{i=1}^n \xi_i\right|\right] \geq \frac{1}{2\sqrt{2}} \cdot  \E\left[\sqrt{\sum_{i=1}^n\xi_i^2}\right].
\end{align*}
\end{theorem}

We can now prove Theorem~\ref{thm:chan-local}. Our proof is similar to that of Chan et al.~\cite{ChanSS12}. The main difference is that instead of defining the notion of ``bad transcripts'' explicitly as in~\cite{ChanSS12}, we account of them implicitly in our averaging argument.

\begin{proof}[Proof of Theorem~\ref{thm:chan-local}]
For convenience, let us denote by $\cR_0$ and $\cR_1$ the distributions of $\bX^0$ and $\bX^1$ respectively. Assume that there is an analyzer that receives the messages from the users (without shuffling), where the $i$th user with input $b_i$ samples $X_i$ from $\cR_{b_i}$ and sends $X_i$ to the analyzer, and output an estimate sum with an expected error at most $\alpha$. We will argue that $\SD(\cR_0, \cR_1) \geq 1 - \Omega\left(\frac{\alpha}{\sqrt{n}}\right)$.

For each message sequence $X_1, \dots, X_n$ where $X_i$ is the message from the $i$th user, we use $A(X_1, \dots, X_n)$ to denote the analyzer's estimate\footnote{Note that we may assume w.l.o.g. that the analyzer is deterministic.} upon receiving these messages. For any input sequence $b_1, \dots, b_n \in \{0, 1\}$, the expected error is
\begin{align*}
\E_{X_1 \sim \cR_{b_1}, \dots, X_n \sim \cR_{b_n}} \left|A(X_1, \dots, X_n) - (b_1 + \cdots + b_n)\right|,
\end{align*}
which must be at most $\alpha$ due to our assumption.

Hence, by averaging over all sequences $b_1, \dots, b_n \in \{0, 1\}$, we have
\begin{align*}
\alpha \geq \E_{b_1, \dots, b_n \sim \{0, 1\}} \E_{X_1 \sim \cR_{b_1}, \dots, X_n \sim \cR_{b_n}} \left|A(X_1, \dots, X_n) - (b_1 + \cdots + b_n)\right|.
\end{align*}
Let us denote the quantity on the right hand side above by $\err$. Furthermore, for each possible message $X \in \supp(\cR_0) \cup \supp(\cR_1)$, let us define the probability distribution $\cF_X$ on $\{0, 1\}$ by $\cF_X(0) = \frac{\cR_0(X)}{\cR_0(X) + \cR_1(X)}$ and $\cF_X(1) = \frac{\cR_1(X)}{\cR_0(X) + \cR_1(X)}$. It is not hard to see that $\err$ can be rearranged as 
\begin{align} \label{eq:average-err-swap}
\err = \E_{X_1, \dots, X_n \sim 0.5\cR_0 + 0.5\cR_1} \E_{b_1 \sim \cF_{X_1}, \dots, b_n \sim \cF_{X_n}} \left|A(X_1, \dots, X_n) - (b_1 + \cdots + b_n)\right|.
\end{align}
Let us now bound the inner expectation as follows.
\begin{align}
&\E_{b_1 \sim \cF_{X_1}, \dots, b_n \sim \cF_{X_n}} \left|A(X_1, \dots, X_n) - (b_1 + \cdots + b_n)\right| \nonumber \\
&= \frac{1}{2}  \E_{b_1, b'_1 \sim \cF_{X_1}, \dots, b_n, b'_n \sim \cF_{X_n}} \left[ \left|A(X_1, \dots, X_n) - (b_1 + \cdots + b_n)\right| + \left|A(X_1, \dots, X_n) - (b'_1 + \cdots + b'_n)\right|\right] \nonumber \\
&\geq  \frac{1}{2}  \E_{b_1, b'_1 \sim \cF_{X_1}, \dots, b_n, b'_n \sim \cF_{X_n}} \left|(b_1 - b'_1) + \cdots + (b_n - b'_n) \right|, \label{eq:triangle-eq}
\end{align}
where the last line follows from triangle inequality. Now, observe that each $(b_i - b'_i)$ is an independent random variable such that
\begin{align*}
b_i - b'_i = 
\begin{cases}
-1 & \text{ with probability } \cF_{X_i}(0)\cF_{X_i}(1), \\
0 & \text{ with probability } 1 - 2\cF_{X_i}(0)\cF_{X_i}(1), \\
1 & \text{ with probability } \cF_{X_i}(0)\cF_{X_i}(1). \\
\end{cases}
\end{align*}
Hence, we may apply the Marcinkiewicz-–Zygmund inequality (Theorem~\ref{thm:anti-concen}), which gives
\begin{align*}
\E_{b_1 \sim \cF_{X_1}, \dots, b_n \sim \cF_{X_n}} \left|A(X_1, \dots, X_n) - (b_1 + \cdots + b_n)\right|
&\geq \frac{1}{4\sqrt{2}} \cdot \E \left[\sqrt{\sum_{i=1}^n (b_i - b'_i)^2}\right] \\
(\text{by power mean inequality}) &\geq \frac{1}{4\sqrt{2}} \cdot \E \left[\frac{\sum_{i=1}^n |b_i - b'_i|}{\sqrt{n}}\right] \\
(\text{by the linearity of expectation}) &= \frac{1}{2\sqrt{2n}} \cdot \sum_{i=1}^n \cF_{X_i}(0) \cF_{X_i}(1)
\end{align*}
Plugging this back into~\eqref{eq:triangle-eq} and using the linearity of expectation once again, we have
\begin{align} \label{eq:error-bound-almost-final}
\err \geq \E_{X_1, \dots, X_n \sim 0.5\cR_0 + 0.5\cR_1} \left[\frac{1}{2\sqrt{2n}} \cdot \sum_{i=1}^n \cF_{X_i}(0)\cF_{X_i}(1)\right]
= \frac{\sqrt{n}}{2\sqrt{2}} \cdot \E_{X \sim 0.5\cR_0 + 0.5\cR_1}\left[\cF_X(0)\cF_X(1)\right].
\end{align}
Finally, we relate the right hand side term with the total variation distance between $\cR_0$ and $\cR_1$ as follows.
\begin{align}
\E_{X \sim 0.5\cR_0 + 0.5\cR_1}\left[\cF_X(0)\cF_X(1)\right]
&= \sum_{X} (0.5 \cR_0(X) + 0.5 \cR_1(X)) \cdot \frac{0.5 \cR_0(X)}{\cR_0(X) + \cR_1(X)} \cdot \frac{0.5 \cR_1(X)}{\cR_0(X) + \cR_1(X)} \nonumber \\
&= \sum_X 0.5 \frac{\cR_0(X)\cR_1(X)}{\cR_0(X) + \cR_1(X)} \nonumber \\
&\geq \sum_X 0.25\min\{\cR_0(X), \cR_1(X)\} \nonumber \\
&= 0.25(1 - \SD(\cR_0, \cR_1)). \label{eq:relating-to-tv}
\end{align}
Combining~\eqref{eq:error-bound-almost-final} and~\eqref{eq:relating-to-tv}, we have $\err \geq  \frac{\sqrt{n}}{8\sqrt{2}} (1 - \SD(\cR_0, \cR_1))$. Since $\err \leq \alpha$, we must have $\SD(\cR_0, \cR_1) \geq 1 - O\left(\frac{\alpha}{\sqrt{n}}\right)$ as desired.
\end{proof}


\section{Proof of Observation~\ref{obs:discrete-gaussian-tail}}
\label{app:discrete-gaussian-tail}

\begin{proof}[Proof of Observation~\ref{obs:discrete-gaussian-tail}]
Let $\ell^*$ be the smallest positive integer such that $\sum_{x \in a\Z \setminus (-\ell^*a, \ell^*a)} \rho_{s^*}(x) \leq e^{-\delta} \lambda \cdot \rho_{s^*}(a\Z)$; such an integer exists because $\rho_{s^*}(a\Z) = \sum_{x \in a\Z} \rho_{s^*}(x) < \infty$.

Consider any $c \in \R$. Let $q = \lfloor c / a\rfloor$ and $r = c - qa$. We may expand $\Pr_{X \sim \cD_{a\Z, s^*, c}}[|X - c| > \ell^* a]$ as
\begin{align*}
&\sum_{x \in a\Z \setminus [c - \ell^* a, c + \ell^* a]} \cD_{a\Z, s^*, c}(x) \\
&= \frac{1}{\rho_{s^*, c}(a\Z)} \left(\sum_{x \in a\Z \setminus [c - \ell^* a, c + \ell^* a]} \rho_{s^*, c}(x) \right) \\
&= \frac{1}{\rho_{s^*, c}(a\Z)} \left(\sum_{x \in a\Z \setminus [c - \ell^* a, c + \ell^* a]} \rho_{s^*}(x - c) \right) \\
&= \frac{1}{\rho_{s^*, c}(a\Z)} \left(\sum_{x \in a\Z \atop x < c - \ell^* a} \rho_{s^*}(x - c) + \sum_{x \in a\Z \atop x > c + \ell^* a} \rho_{s^*}(x - c)\right) \\
&\leq \frac{1}{\rho_{s^*, c}(a\Z)} \left(\sum_{x \in a\Z \atop x < c - \ell^* a} \rho_{s^*}(x - (q - 1)a) + \sum_{x \in a\Z \atop x > c + \ell^* a} \rho_{s^*}(x - qa)\right) \\
&= \frac{1}{\rho_{s^*, c}(a\Z)} \left(\sum_{x \in a\Z \atop x < c - (q - 1)a - \ell^* a} \rho_{s^*}(x) + \sum_{x \in a\Z \atop x > c - qa + \ell^* a} \rho_{s^*}(x)\right) \\
&\leq \frac{1}{\rho_{s^*, c}(a\Z)} \left(\sum_{x \in a\Z \atop x \leq - \ell^* a} \rho_{s^*}(x) + \sum_{x \in a\Z \atop x \geq \ell^* a} \rho_{s^*}(x)\right) \\
&= \frac{1}{\rho_{s^*, c}(a\Z)} \cdot \sum_{x \in a\Z \setminus (-\ell^*a, \ell^*a)} \rho_{s^*}(x) \\
&\leq \frac{\rho_{s^*}(a\Z) e^{-\delta} \lambda}{\rho_{s^*, c}(a\Z)},
\end{align*}
where the last inequality follows from our choice of $\ell^*$. Finally, recall from Lemma~\ref{lem:discrete-gaussian-shift} that $\rho_{s^*, c}(a\Z) \geq e^{-\delta} \cdot \rho_{s^*}(a\Z)$. Plugging this back into the above inequality yields the desired claim.
\end{proof}

\end{document}